\documentclass[oribibl]{llncs}
\usepackage{llncsdoc}
\usepackage{pfFormLNCSV11}
\usepackage{pfResV9}
\usepackage{url}

\title{The Power of Proofs: New Algorithms for Timed Automata Model Checking (with Appendix) \thanks{Research supported by NSF grant CCF-0926194. This paper is the preprint of the FORMATS 2014 version of the paper, and this preprint is the supplement that contains the full Appendix. The final version \cite{fontana_power_2014} is published by Springer, and is available at http://link.springer.com/chapter/10.1007/978-3-319-10512-3\_9 DOI: 10.1007/978-3-319-10512-3\_9.}}
\author{Peter Fontana and Rance Cleaveland}
\institute{University of Maryland, College Park, Department of Computer Science}
\date{\today}

\usepackage[table]{xcolor}

\graphicspath{{Figures/}}

\begin{document}

\maketitle

\begin{abstract}
This paper presents the first model-checking algorithm for an expressive modal mu-calculus over timed automata, $L^{\mathit{rel}, \mathit{af}}_{\nu,\mu}$, and reports performance results for an implementation. This mu-calculus contains extended time-modality operators and can express all of TCTL. Our algorithmic approach uses an ``on-the-fly" strategy based on proof search as a means of ensuring high performance for both positive and negative answers to model-checking questions.  In particular, a set of proof rules for solving model-checking problems are given and proved sound and complete; our algorithm then model-checks a property by constructing a proof (or showing none exists) using these rules.  One noteworthy aspect of our technique is that we show that verification performance can be improved with \emph{derived rules}, whose correctness can be inferred from the more primitive rules on which they are based. In this paper, we give the basic proof rules underlying our method, describe derived proof rules to improve performance, and we compare our implementation to UPPAAL.
\end{abstract}

\section{Introduction}

Timed automata are used to model real-time systems in which time is continuous and timing constraints may refer to elapsed time between system events \cite{alur-a-theory-1994}.  The timed automata model provides a balance between expressiveness and tractability: a variety of different real-time systems can be captured in the formalism, and various properties, including safety (reachability) and liveness, can also be decided automatically for a given automaton \cite{aceto-is-your-2002,alur-timed-automata-1999,alur-model-checking-in-1993}. 

To specify these properties, different logics have been devised.  One popular logic, Timed Computation Tree Logic (TCTL) \cite{alur-model-checking-in-1993}, extends the untimed Computation Tree Logic (CTL) \cite{clarke-automatic-verification-1986}  by adding time constraints to the modal operators.  Other researchers explored timed extensions to the modal mu-calculus \cite{emerson-efficient-model-1986}. One such extension, called $T_{\mu}$ \cite{henzinger-symbolic-model-1994} extends the untimed modal mu-calculus with a single-step operator. Another extension, which we refer to as $L_{\nu,\mu}$~\cite{sokolsky-local-model-1995,zhang-fast-generic-2005,zhang-fast-on-the-fly-2005}, extends the modal mu-calculus with separate time and action modal operators. This logic is sufficient for expressing some basic safety and liveness properties.  However, it cannot express all of TCTL \cite{fontana-expressiveness-results-2014}.  To address this, $L_{\nu,\mu}$ was extended with \emph{relativization operators} by \cite{bouyer-timed-modal-2011}; we denote this logic as $L^{rel}_{\nu,\mu}$. These additional operators make the logic expressive enough to express all of TCTL \cite{fontana-expressiveness-results-2014}. (Bouyer et al.\ \cite{bouyer-timed-modal-2011} included only greatest fixpoints, yielding $L_{\nu}$, which they referred to as $L_{c}$; the least fixpoints in $L^{rel}_{\nu,\mu}$ not in $L^{rel}_{\nu}$ add expressive power \cite{fontana-expressiveness-results-2014}.) 

Over the model of timed automata, the model checking problem for  $L_{\nu,\mu}$ is EXPTIME-complete \cite{aceto-is-your-2002}. Bouyer et al.\ \cite{bouyer-timed-modal-2011} show that formulas using the relativization operators can be model-checked in EXPTIME. Hence, model checking  $L^{rel}_{\nu,\mu}$ over timed automata is EXPTIME-complete. The same model-checking problem for TCTL over timed automata is PSPACE-complete \cite{alur-model-checking-in-1993}. 

While timed logics were being studied, tools and implementation algorithms were developed as well. Much of the development focused on handling subsets of properties specified in TCTL. A widely-used tool, UPPAAL \cite{behrmann-a-tutorial-2004}, supports a fragment of TCTL, which includes many safety and liveness properties; other tools, including KRONOS \cite{yovine-kronos:-a-1997}, Synthia \cite{peter-synthia:-verification-2011}, and RED/REDLIB \cite{wang-redlib-for-2006}, have also been developed, some of which are able to model-check all of TCTL.  Additionally, some tools were developed for timed modal-mu calculi. Two tools that can model check fragments of a timed mu-calculus include CMC \cite{laroussinie-cmc:-a-1998}, which can handle $L_{\nu}$, and CWB-RT \cite{fontana-data-structure-2011,zhang-fast-generic-2005,zhang-fast-on-the-fly-2005}, which can check safety properties written in $L_{\nu}$.

The contributions of this paper include the first algorithm, and an implementation, to model check $L^{\mathit{rel}, \mathit{af}}_{\nu,\mu}$. By definition, $L^{\mathit{rel}, \mathit{af}}_{\nu,\mu}$ consists of the so-called \emph{alternation-free} formulas of $L^{rel}_{\nu,\mu}$ and is thus a superset of $L^{rel}_{\nu}$. Assuming non-zeno and timelock-free automata,  $L^{\mathit{rel}, \mathit{af}}_{\nu,\mu}$ is strong enough to express all of TCTL \cite{fontana-expressiveness-results-2014}. Our implementation extends the tool CWB-RT \cite{fontana-data-structure-2011,zhang-fast-generic-2005,zhang-fast-on-the-fly-2005}. Implementation details of the model checker are discussed in Section \ref{s:mcalg}; in Section \ref{s:examples}, we give a demonstration of some models and properties that can be model checked by our tool as well as a performance comparison to UPPAAL. 

CWB-RT is a proof-search model checker: it verifies properties by constructing a proof using a set of proof rules. These proof rules decompose the given goal (does the automaton satisfy a formula) into (smaller) subgoals.  These proof search methods were used for the untimed modal mu-calculus in \cite{cleaveland-tableau-based-model-1990}, explored in \cite{sokolsky-local-model-1995}, and extended to the timed setting in \cite{zhang-fast-generic-2005,zhang-fast-on-the-fly-2005} in order to produce a fast on-the-fly model checker that can model check timed automata incrementally. The generated proofs not only give additional correctness information but also can be used as a mechanism to improve model-checking performance. We develop the additional proof rules to check the relativized operators, extending the proof rules used in \cite{zhang-fast-generic-2005,zhang-fast-on-the-fly-2005}. The additional rules are discussed in Section \ref{s:proofrules}.

Furthermore, through select \emph{derived} proof rules, we can enhance performance. These derived rules, together with a judicious use of \emph{memoization}, yield dramatic performance improvements. We discuss the derived proof rules in Section \ref{s:optderiv}.

\section{Background}
\label{s:back}

\subsection{Timed Automata}

This section defines the syntax of timed automata and sketches their semantics. The interested reader is referred to \cite{alur-timed-automata-1999,fontana-a-menagerie-2014} for a fuller account. To begin with, timed automata rely on \emph{clock constraints}.

\begin{definition}[Clock constraint $cc \in \Phi(CX)$] Given a nonempty finite set of clocks $CX = \mset{x_1, x_2, \ldots, x_n}$ and $d \in \Zz^{\geq 0}$ (a non-negative integer), a \emph{clock constraint $cc$} may be constructed using the following grammar:
\begin{equation*}
cc::= x_i < d \ | \ x_i \leq d \ | \ x_i > d \ | \ x_i \geq d \ | \ cc \lgcand cc
\end{equation*}
$\Phi(CX)$ is the set of all possible clock constraints over $CX$. We also use the following abbreviations: true (\lgtrue) for $x_1 \geq 0$, false (\lgfalse) for $x_1 < 0$, and $x_i = d$ for $x_i \leq d \lgcand x_i \geq d$.
\label{def:cxcons}
\end{definition}

\noindent
Timed automata may now be defined as follows.

\begin{definition}[Timed automaton]
A \emph{timed automaton} is a tuple \\ $(L, l_0, \Sigma, CX, I, E)$, where:
\begin{itemize}
\item $L$ is the finite set of \emph{locations}. 
\item $l_0 \in L$ is the \emph{initial location}.
\item $\Sigma$ is the finite set of \emph{action symbols}. 
\item $CX = \mset{x_1, x_2, \ldots, x_n}$ is the nonempty finite set of \emph{clocks}.
\item $\nmfunc{I}{L}{\Phi(CX)}$ maps each location $l$ to a clock constraint, $I(l)$, referred to as the \emph{invariant} of $l$.
\item $E \subseteq L \times \Sigma \times \Phi(CX) \times 2^{CX} \times L$ is the set of \emph{edges}.  In an edge $e = (l, a, cc, \lambda, l')$ from $l$ to $l'$ with action $a$, $cc \in \Phi(CX)$ is the \emph{guard} of $e$, and $\lambda$ represents the set of clocks to \emph{reset} to $0$.
\end{itemize}
\label{def:timedaut}
\end{definition}

\noindent
The semantics of timed automata rely on \emph{clock valuations}, which are functions $\nmfunc{\nu}{CX}{\Rr^{\geq 0}}$ ($\Rr^{\geq 0}$ is the set of non-negative real numbers); intuitively, $\nu(x_i)$ is the current time value of clock $x_i$.  A timed automaton begins execution in its initial location with the initial clock valuation $\nu_0$ assigning $0$ to every clock.  When the automaton is in a given clock location $l$ with current clock valuation $\nu$, two types of transitions can occur: time advances and action executions. During a time advance, the location stays the same and the clock valuation $\nu$ advances $\delta \in \Rr^{\geq 0}$ units to the valuation $\nu + \delta$, where $\nu + \delta$ is defined as $(\nu + \delta)(x_i) = \nu(x_i) + \delta$. For a time advance to be allowed, for all $0 \leq \delta' \leq \delta$, $\nu + \delta'$ must satisfy the invariant of location $l$. Due to convexity of clock constraints, it suffices to ensure that both $\nu$ and $\nu + \delta$ satisfy $I(l)$. An \emph{action execution} of action $a$ can occur when $\nu$ satisfies the guard for an edge leading from $l$ to $l'$, the edge is labeled by action $a$, and , the invariant of $l'$ is satisfied after the clocks are reset as specified in the edge. In this case the location changes to $l'$ and the clocks in $\lambda$ are reset to $0$.  These intuitions can be formalized as a labeled transition system whose states consist of locations paired with clock valuations, each state notated as $(l,\nu)$.  A \emph{timed run} of the automaton is a sequence of transitions starting from the initial location and $\nu_0$. On occasion, we also augment each timed automaton with a set of atomic propositions $AP$ and a labeling function $\nmfunc{M}{L}{\pwmset{AP}}$ where $M(l)$ is the subset of propositions in $AP$ that location $l$ satisfies.

\begin{example}[Train timed automaton] The timed automaton in \mbox{Figure \ref{fig:ta1}} models a train component of the GRC (Generalized Railroad Crossing) protocol \cite{heitmeyer-the-generalized-1994}. There are three locations: 0:\:far (initial), 1:\:near, and 2:\:in; and one clock $x_1$.  $\Sigma$ has the actions $approach$, $in$, and $exit$. Here, location 1:\:near has the invariant $x_1 \leq 4$ while 0:\:far has the vacuous invariant $\lgtrue$.  The edge $($1:\:near$, in, x_1 = 4, \mset{x_1},  $2:\:in$)$ has action $in$, guard $x_1 = 4$, and resets $x_1$ to 0.

\begin{figure}[t]

\centerline{\includegraphics[scale=1.0]{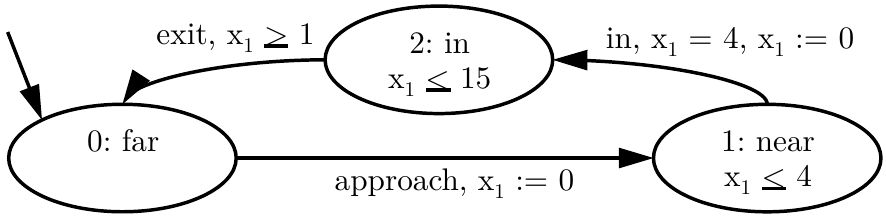}}
\caption{Timed automaton of a train.}
\label{fig:ta1}
\end{figure}

A sample timed run of this timed automaton is: $
(\text{0:\:far}, x_1 = 0) \ttrans{5}  (\text{0:\:far}, x_1 = 5) \ttrans{approach} (\text{1:\:near}, x_1 = 0)
\ttrans{4} (\text{1:\:near}, x_1 = 4) \ttrans{in} (\text{2:\:in}, x_1 = 0) \ttrans{3} (\text{2:\:in}, x_1 = 3)
 \ttrans{2} (\text{2:\:in}, x_1 = 5) \ttrans{exit}  (\text{0:\:far}, x_1 = 5) \ldots $
\label{ex:ta}
\end{example}

\subsection{Timed Logic $L^{rel}_{\nu,\mu}$ and Modal Equation Systems (MES)}

The following definition of $L^{rel}_{\nu,\mu}$ uses the modal-equation system (MES) format used in \cite{cleaveland-a-linear-time-1993} for untimed systems and in \cite{zhang-fast-generic-2005,zhang-fast-on-the-fly-2005} for $L_{\nu,\mu}$. 

\begin{definition}[$L_{\nu,\mu}$, $L^{rel}_{\nu,\mu}$ basic formula syntax]
Let $CX = \mset{x_1, x_2, \ldots}$ and $CX_f = \mset{z,z_1,\ldots}$ be disjoint sets of clocks.  Then
the syntax of a $L_{\nu,\mu}$ \emph{basic formulas} is given by the following grammar:
\begin{align*}
\psi ::= & p \barsep \lgnot{p} \barsep \lgtrue \barsep \lgfalse \barsep cc \barsep Y \barsep \psi \lgcand \psi \barsep \psi \ \lgcor \psi  \barsep \mmEXra{a}{\psi} \\ 
&\barsep \mmAXra{a}{\psi} \barsep \tfoEsh{\psi} \barsep \tfoAsh{\psi} \barsep \tfreeze{z}{\psi} 
\end{align*}
Here, $p \in AP$ is an atomic proposition, $cc \in \Phi(CX)$ is a clock constraint over clock set $CX$, $Y \in Var$ is a propositional variable ($Var$ is the set of propositional variables), and $a \in \Sigma$ is an action. In formula $z.\psi$ $z$ is a clock in $CX_{f}$; the $z.$ operator is often referred to as \emph{freeze quantification}.

The \emph{relativized timed modal-mu calculus $L^{rel}_{\nu,\mu}$ syntax} replaces $\tfoEsh{\psi}$ and $\tfoAsh{\psi}$ with $\tfoEr{\psi_1}{\psi_2}$ and $\tfoAr{\psi_1}{\psi_2}$, where each $\psi_1$ and $\psi_2$ are basic formulas in $L^{rel}_{\nu,\mu}$.
\label{def:lnumusyn}
\end{definition}

What follows is a sketch of the semantics; \cite{bouyer-timed-modal-2011,fontana-expressiveness-results-2014} contains a formal definition. Formulas are interpreted with respect to states (i.e. (location, clock valuation) pairs) of a timed automaton whose clock set is $CX$ and labeling function is $M$, and an environment $\theta$ associating each propositional variable $Y$ with a set of states.  A state $(l,\nu)$ satisfies an atomic proposition $p$ if and only if $p$ is in the set $M(l)$.  A state satisfies $Y$ if and only if $(l,\nu) \in \theta(Y)$.  $\mmEXra{a}{\psi}$ holds in a state if, after executing action $a$, $\psi$ is true of the state after the action transition; $\mmAXra{a}{\psi}$ means after all action transitions involving $a$, $\psi$ holds in the target state; $\tfoEsh{\psi}$ holds of a state if after some time advance of $\delta \geq 0$, $\psi$ holds in the new state, while $\tfoAsh{\psi}$ is satisfied in a state if for all possible time advances of $\delta \geq 0$, $\psi$ is true in the resulting states. Formula $\tfreeze{z}{\psi}$ holds in a state if, after introducing a new clock $z$ ($z$ is not a clock of the timed automaton) and setting it to $0$ without altering other clocks, $\psi$ is true. The formula $\tfoEr{\psi_1}{\psi_2}$ means, ``there exists a time advance where $\psi_2$ is true and $\psi_1$ is true for all times up to, but not including, that advance'', and $\tfoAr{\psi_1}{\psi_2}$ means, ``either $\psi_2$ is true for all time advances or $\psi_1$ releases $\psi_2$ from being true after some time advance.''  

We also introduce two derived operators: $\displaystyle \mmAX{\psi}$ for $\bigwedge_{a \in \Sigma_{TA}}\mmAXra{a}{\psi}$ (for all next actions) and $\displaystyle \mmEX{\psi}$ for  $\bigvee_{a \in \Sigma_{TA}}\mmEXra{a}{\psi} $ (there exists a next action). It may be seen that $\tfoEsh{\psi}$ is equivalent to $\tfoEr{\lgtrue}{\psi}$, and $\tfoAsh{\psi}$ to $\tfoAr{\lgfalse}{\psi}$.

$L^{rel}_{\nu,\mu}$ MESs are mutually recursive systems of equations whose right-hand sides are basic formulas as specified above.  The formal definition follows.

\begin{definition}[$L^{rel}_{\nu,\mu}$ MES syntax] Let $X_1, X_2, \ldots, X_{v}$ be propositional variables, and let $\psi_1, \ldots \psi_{v}$ all be $L^{rel}_{\nu,\mu}$ basic formulae. Then \emph{a $L^{rel}_{\nu,\mu}$ modal equation system (MES)} is an ordered system of equations as follows, where each equation is labeled with a \emph{parity} ($\mu$ for least fixpoint, $\nu$ for greatest fixpoint): $X_1 \overset{\mu/\nu}{=} \psi_1, X_2 \overset{\mu/\nu}{=} \psi_2, \ldots, X_v \overset{\mu/\nu}{=} \psi_v$.

In our MES, we will assume that all variables are \emph{bound} (every variable in the right of the equation appears as some left-hand variable).
\label{def:lnymumessyn}
\end{definition}

The formal definition of the semantics of MESs may be found in~\cite{zhang-fast-generic-2005,zhang-fast-on-the-fly-2005}; we recount the highlights here.  Given a timed automaton and atomic-proposition interpretation $M$, a basic $L^{rel}_{\nu,\mu}$ formula may be seen as a function mapping sets of timed-automaton states (corresponding to the meaning of the propositional variables to the formula) to a single set of states (the states that make the formula true, given the input sets just referred to).  The set of subsets of timed-automaton states ordered by set inclusion form a complete lattice; it turns out that the functions over this lattice definable by basic formulae are monotonic over this lattice, meaning they have unique greatest and least fixpoints.  This fact is the lynch-pin of the formal semantics of MESs.  Specifically, given MES $X_1 \overset{\mu/\nu}{=} \psi_1, \ldots, X_v \overset{\mu/\nu}{=} \psi_v$, we may construct  a function that, given a set of states for $X_1$, returns the set of states satisfying $\psi_1$, where the values for $X_2, \ldots, X_v$ have been computed recursively.  This function is monotonic, and therefore has a unique least and greatest fixpoint.  If the parity for $X_1$ is $\mu$, then the set of states satisfying $X_1$ is the least fixpoint of this function, while if the parity is $\nu$ then the set of states satisfying $X_1$ is the greatest fixpoint.  By convention, the meaning of a MES is the set of states associated with $X_1$, the first left-hand-side in the sequence of equations. However, in the MES, each variable $X_i$ can be interpreted as its own subformula; this interpretation will prove useful constructing proofs that a state satisfies a MES.

Given timed automaton $TA$, atomic-proposition interpretation function $M$, and propositional variable environment $\theta$, we use $\satset{\psi}_{TA,M,\theta}$ to denote the set of states satisfying $\psi$.  For an MES $\mathcal{M}$ of form $X_1 \overset{\mu/\nu}{=} \psi_1 \ldots X_v \overset{\mu/\nu}{=} \psi_v$, we write $\satset{\mathcal{M}}_{TA,M, \theta}$, or equivalently $\satset{X_1}_{TA,M, \theta}$ when there is no confusion, for the set of states satisfying the MES.

To handle the clocks used in freeze quantification ($\tfreeze{z}{\psi}$), we extend the timed automaton's states $(l,\nu)$ to \emph{extended states}   $(l,\nu, \nu_{f})$ using the additional valuation component $\nmfunc{\nu_{f}}{CX_{f}}{\Rr^{\geq{0}}}$. This formalism comes from \cite{bouyer-timed-modal-2011}. When clear from context, we will refer to an extended state as $(l,\nu)$ and omit the explicit notation of $\nu_{f}$.

In this paper we only consider MESs that are \emph{alternation-free}. Intuitively, an MES is alternation free if there is no mutual recursion involving variables of different parities.  For more information on the notion, see \cite{emerson-efficient-model-1986}. We denote the alternation-free fragment of $L^{rel}_{\nu,\mu}$ as $L^{\mathit{rel}, \mathit{af}}_{\nu,\mu}$. By definition, $L^{\mathit{rel}, \mathit{af}}_{\nu,\mu}$ is a superset of $L^{rel}_{\nu}$ because any formula with an alternation must have at least one greatest fixpoint and at least one least fixpoint. The alternation-free restriction is not prohibitive because for any timelock-free nonzeno timed automaton (see \cite{bowman-how-to-2006}), we can express any TCTL formula into a $L^{\mathit{rel}, \mathit{af}}_{\nu,\mu}$ MES \cite{fontana-expressiveness-results-2014}. 

\begin{example}[Specifying properties with MES]
Again consider the timed automaton in Figure \ref{fig:ta1} of Example \ref{ex:ta}. Two $L^{\mathit{rel}, \mathit{af}}_{\nu,\mu}$ specifications we can ask are:
\begin{align}
X_1 \mesnueq& \neg broke \lgcand \tfoAsh{\mmAX{X_1}} \label{eq:a1} \\
X_1 \mesnueq& \neg far \lgcor \Bigl(\tfoAsh{\mmAX{X_1}} \lgcand \tfoEsh{\tfreeze{z}{\tfoAsh{z < 1}}}\Bigr) \label{eq:a2}
\end{align}
Equation \ref{eq:a1} says ``it is always the case the the train is not broken,'' and equation \ref{eq:a2} says ``it is inevitable that a train is not far.''
\label{ex:spec1}
\end{example}

\section{Checking $L^{\mathit{rel}, \mathit{af}}_{\nu,\mu}$ Properties: A Proof-Based Approach}
\label{s:proofrules}

The $L^{\mathit{rel}, \mathit{af}}_{\nu,\mu}$ model-checking problem for timed automata may be specified as follows:  given timed automaton $TA = (L, l_0, \Sigma_{TA}, CX, I, E)$, atomic-proposition interpretation function $M$, and $L^{\mathit{rel}, \mathit{af}}_{\nu,\mu}$ formula $\psi$ with initial environment $\theta$, determine if the initial state of $TA$ satisfies $\psi$, i.e.: is $(l,\nu) \in \satset{\psi}_{TA,M,\theta}$.  This section describes the proof-based approach that we use to solve such problems.

Our model-checking technique relies on the construction of proofs that are intended to establish the truth of judgments, or \emph{sequents}, of the form
$\sqjudgement{(l, cc)}{\psi},$
where $l \in L$ is a location, $cc \in \Phi(CX \cup CX_{f})$ is a clock constraint, and $\psi$ is a $L^{\mathit{rel}, \mathit{af}}_{\nu,\mu}$ formula. Note that $cc$ includes clocks from the timed automaton as well as any clocks used in freeze quantifications.  Note that semantically, a clock constraint $cc$ can be viewed as the set of valuations $cc = \mset{\nu \barsep \nu \models cc}$; likewise, we can encode a valuation $\nu$ as the clock constraint $cc_{\nu} = x_1 = \nu(x_1) \lgcand \ldots \lgcand x_n = \nu(x_n)$. A \emph{proof rule} contains a finite number of hypothesis sequents and a conclusion sequent and may be written as follows.
\begin{center}
\begin{proofinlinepf}
\AxiomC{Premise 1}
\AxiomC{$\ldots$}
\AxiomC{Premise $n$}
\RightLabel{(\emph{Rule Name})}
\TrinaryInfC{Conclusion}
\end{proofinlinepf} 
\end{center}
The intended reading of such a rule is that if each premise is valid, then so is the conclusion. Some proof rules, \emph{axioms}, have no premises and thus assert the truth the validity of their conclusion. Given a collection of rules, our verifier builds a \emph{proof} by chaining these proof rules together. A proof is \emph{valid} if the proof rules are applied properly, meaning that the premise of the previous rule is the conclusion of the next rule. The proof rules are designed to be sound and complete, meaning: $(l,\nu) \in \satset{\psi}_{TA,M,\theta}$ if and only if there is a valid proof for $(l, cc_{\nu}) \vdash \psi$.  The proof-construction process proceeds in an ``on-the-fly" manner:  rules whose conclusion matches the sequent to be proved are applied to this goal sequent, yielding new sequents that must be proved.  This procedure is applied recursively, and systematically, until either a proof is found, or none can be.

\subsection{Proof Rules for $L^{af}_{\nu,\mu}$ Over Timed Automata}

The proof-based approach in this paper is inspired by a generic proof framework in \cite{zhang-fast-generic-2005,zhang-fast-on-the-fly-2005} based on a general theory called Predicate Equation Systems (PES). PES involved fixpoint equations over first-order predicates and used the proof-search to  establish the validity of a PES. For practical reasons, one generally wishes to avoid the construction of the PES explicitly; this paper adopts this point of view, and the proof rules that it presents thus involve explicit mention of timed-automata notions, including location and edge.
A selection of proof rules derived from~\cite{zhang-fast-generic-2005,zhang-fast-on-the-fly-2005} is given in Figure \ref{fig:selectpr}. The remaining rules are in Appendix A. 
 Several comments are in order.
\begin{enumerate}
\item
Each rule is intended to relate a conclusion sequent involving a formula with a specific outermost operator to premise sequents involving the maximal subformula(e) of this formula.  The name of the rule is based on this operator.
\item
The premises also involve the use of functions $succ$ and $pred$.  Intuitively, $succ((l,cc))$ represents all states that are time successors of any state whose location component is $l$ and whose clock valuation satisfies $cc$, while $pred((l,cc))$ are the time predecessors of these same states. These operators may be computed symbolically; that is, for any $(l, cc)$ there is a $cc'$ such that $(l,cc')$ is equivalent to $succ((l,cc))$.  
\item
Some of the rules involve \emph{placeholders}, which are (potentially) unions of clock constraints, given as (subscripted versions of) $\phi$.  Given a specific placeholder, the premise sequent $(l,cc), \phi$ is semantically equivalent to $(l, cc \lgcand \phi)$; however, for notational and implementation ease, the placeholder $\phi$ is tracked separately from the clock constraint $cc$. 
\end{enumerate}

\begin{figure}[t]
\begin{proofinlinepf}
\AxiomC{\sqjudgement{(l_1,cc \lgcand g_1)}{\psi[\lambda_1 := 0]}}
\AxiomC{\ldots}
\AxiomC{\sqjudgement{(l_n,cc \lgcand g_n)}{\psi[\lambda_n := 0]}}
\RightLabel{($\mmAXraop{a}_{Act}$), cond$\mmAXraop{a}$}
\TrinaryInfC{\sqjudgement{\sqpremta}{\mmAXra{a}{\psi}}}
\end{proofinlinepf} \parspc

cond$\mmAXraop{a}$: $\bigcup_{i}\mset{(g_i,\lambda_i,l_i)} = \mset{(l',g',\lambda') \barsep (l,a,g',\lambda',l') \in E}$ \parspc

\begin{proofinlinepf}
\AxiomC{\sqjudgement{\sqpremta, \phi_s}{\psi_1}}
\AxiomC{\sqjudgement{\sqpremta, \lgcnot \phi_{s}}{\psi_2}}
\RightLabel{($\vee_{c}$)}
\BinaryInfC{\sqjudgement{\sqpremta}{\psi_1 \lgcor \psi_2}}
\end{proofinlinepf} 
\begin{proofinlinepf}
\AxiomC{\sqjudgement{succ(\sqpremta)}{\psi}}
\RightLabel{($\forall_{t1}$)}
\UnaryInfC{\sqjudgement{\sqpremta}{\tfoAsh{\psi}}}
\end{proofinlinepf} \parspc

\begin{proofinlinepf}
\AxiomC{\sqjudgement{succ(\sqpremta), \phi_{s}}{\psi}}
\AxiomC{\sqjudgement{succ(\sqpremta, \phi_{\forall})}{succ(\sqpremta) \lgcand \phi_{s}}}
\RightLabel{($\forall_{t2}$)}
\BinaryInfC{\sqjudgement{\sqpremta, \phi_{\forall}}{\tfoAsh{\psi}}}
\end{proofinlinepf} \parspc

\begin{proofinlinepf}
\AxiomC{\sqjudgement{succ(\sqpremta), \phi_{s}}{\psi}}
\AxiomC{\sqjudgement{\sqpremta}{pred(\phi_{s})}}
\RightLabel{($\exists_{t1}$)}
\BinaryInfC{\sqjudgement{\sqpremta}{\tfoEsh{\psi}}}
\end{proofinlinepf}
\begin{proofinlinepf}
\AxiomC{\sqjudgement{succ(\sqpremta), \phi_{s}}{\psi}}
\AxiomC{\sqjudgement{\phi_{\exists}}{pred(\phi_{s})}}
\RightLabel{($\exists_{t2}$)}
\BinaryInfC{\sqjudgement{\sqpremta, \phi_{\exists}}{\tfoEsh{\psi}}}
\end{proofinlinepf}  
\caption{Select proof rules from \cite{zhang-fast-generic-2005,zhang-fast-on-the-fly-2005} adapted for timed automata and MES.}
\label{fig:selectpr}
\end{figure}

More discussion of placeholders is in order.  Intuitively, placeholders encode a set of clock valuations that will make a sequent valid, and which will be computed once the proof is complete.  In practice, we are interested in computing the largest such set. To understand their use in practice, consider the operator $\exists$. To check $\exists$, we need to find some time advance $\delta$ such that $\psi$ is satisfied after $\delta$ time units. Rather than non-deterministically guessing $\delta$, we use a placeholder $\phi_{s}$ in the left premise in rule $\exists_{t1}$ to encode all the time valuations that ensure satisfaction of $\psi$. The right premise then checks that the placeholder $\phi_s$ is some $\delta$-unit time elapse from $(l,cc)$. The placeholder allows us to delay the non-deterministic guess of the value of $\phi_s$ until it is no longer required to guess. Additionally, for performance reasons, we use \emph{new placeholders} to handle time advance operators for sequents with placeholders.  An example may be found in Rule $\exists_{t2}$, where a new placeholder $\phi_\exists$ is introduced in the right premise.  While useful for performance, this choice results in subtle implementation complexities, which we discuss in Section \ref{ss:placecomplex}.

\textbf{Constructing Proofs.} Given sequents and proof rules, proofs now may be constructed in a goal-directed fashion. A sequent is proven by applying a proof rule whose conclusion matches the form of that sequent, yielding as subgoals the corresponding premises of that rule.  These subgoals may then recursively be proved.  If a sequent may be proved using a rule with no premises, then the proof is complete; similarly, if a sequent is encountered a second time (because of loops in the timed automaton and recursion in an MES), then the second occurrence is also a leaf.  Details may be found in~\cite{zhang-fast-generic-2005,zhang-fast-on-the-fly-2005}.  If the recurrent leaf involves an MES variable with parity $\mu$, then the leaf is unsuccessful; if it involves a variable with parity $\nu$, it is successful.  A proof is valid if all its leaves are successful.

\begin{example} To illustrate the proof rules, consider the timed automaton in Figure \ref{fig:ta1}. Suppose we wish to prove the sequent $(2: in, x_1 \leq 3) \vdash \mmAXra{exit}{0: far}$.  Utilizing the first proof rule in Figure \ref{fig:selectpr}, we get the proof:
\begin{prooftreepf}
\AxiomC{}
\UnaryInfC{\sqjudgement{(0: far, 1 \leq x_1 \leq 3)}{0:far}}
\UnaryInfC{\sqjudgement{(2: in, x_1 \leq 3) }{\mmAXra{exit}{0: far}}}
\end{prooftreepf} 
In this rule, we intersect the clock constraint with the guard $x_1 \geq 1$ (if $x_1 < 1$, then there are no possible actions so the formula is true), make the destination location the new sequent, and ask if the destination satisfies the formula. Since the location is $0: far$, the proof is complete.
\end{example}

\subsection{New Proof Rules for the Relativized Operators of $L^{\mathit{rel}, \mathit{af}}_{\nu,\mu}$}
\label{ss:newrules}

We now introduce rules for handling the relativized time-passage modalities in $L_{\nu,\mu}^{\mathit{rel}, \mathit{af}}$. Figure \ref{fig:newrules} gives the rules for the operator $\tfoEr{\psi_1}{\psi_2}$. For the $\tfoAr{\psi_1}{\psi_2}$ operator, we use the derivation given in Lemma \ref{lem:deriv}. 

\begin{figure}[tb]
\begin{proofinlinepf}
\AxiomC{\sqjudgement{\sqpremta, \phi_{s_1}}{\psi_1}}
\AxiomC{\sqjudgement{\sqpremta, \phi_{s_2}}{\psi_2}}
\AxiomC{\sqjudgement{\sqpremta}{\phi_{s_1} \lgcor \phi_{s_2}}}
\RightLabel{($\vee_{s}$)}
\TrinaryInfC{\sqjudgement{\sqpremta}{\psi_1 \lgcor \psi_2}}
\end{proofinlinepf} \parspc 

\begin{proofinlinepf}
\AxiomC{\sqjudgement{\sqpremta, \phi_{s_1}}{\psi_1}}
\AxiomC{\sqjudgement{\sqpremta, \phi_{s_2}}{\psi_2}}
\AxiomC{\sqjudgement{\sqpremta, \phi_{\vee}}{\phi_{s_1} \lgcor \phi_{s_2}}}
\RightLabel{($\vee_{s2}$)}
\TrinaryInfC{\sqjudgement{\sqpremta, \phi_{\vee}}{\psi_1 \lgcor \psi_2}}
\end{proofinlinepf} \parspc

\begin{proofinlinepf}
\AxiomC{\sqjudgement{succ(\sqpremta), \phi_{s}}{\psi_{2}}}
\AxiomC{\sqjudgement{succ(\sqpremta) ,pred_{<}(\phi_{s})}{\psi_{1}}}
\AxiomC{\sqjudgement{\sqpremta}{pred(\phi_{s})}}
\RightLabel{($\exists_{r1}$)}
\TrinaryInfC{\sqjudgement{\sqpremta}{\tfoEr{\psi_1}{\psi_2}}}
\end{proofinlinepf} \parspc

\begin{proofinlinepf}
\AxiomC{\sqjudgement{succ(\sqpremta), \phi_{s'}}{\psi_{2}}}
\AxiomC{\sqjudgement{succ(\sqpremta) ,pred_{<}(\phi_{s'})}{\psi_{1}}}
\AxiomC{\sqjudgement{\sqpremta, \phi_{s}}{pred(\phi_{s'})}}
\RightLabel{($\exists_{r2}$)}
\TrinaryInfC{\sqjudgement{\sqpremta, \phi_{s}}{\tfoEr{\psi_1}{\psi_2}}}
\end{proofinlinepf} 

\caption{Proof Rules for $\vee$ and $\tfoEr{\phi_1}{\phi_2}$.}
\label{fig:newrules} 
\end{figure}

Here is an explanation of the proof rule $\exists_{r1}$; the proof rule $\exists_{r2}$ is similar. The idea is for the placeholder $\phi_{s}$ to encode the $\delta$ time advance needed for $\psi_1$ to be true. The proof-rule premises enforce that this placeholder has three properties:
\begin{enumerate}
	\item \emph{Left premise:} This premise checks that after the time advance taken by $\phi_{s}$, $\psi_2$ is satisfied.
	\item \emph{Middle premise:} This premise checks that until all $\delta$ time-units have elapsed, that $\psi_1$ is indeed true. The $pred_{<}(\phi_s)$ encodes the times before $\phi_s$.
	\item \emph{Right premise:} This premise checks that $\phi_s$ encodes some range of time elapses $\delta$, ensuring that the state can elapse to valuations in $\phi_s$.
\end{enumerate}
To implement this rule, we check the premises in left-to-right order. Some subtleties involving the middle premise are discussed in Section \ref{ss:placecomplex}.

Now we give the claims ensuring the correctness of these new proof rules. Their proofs are in Appendix B. This first lemma is a corrected version of a similar lemma in \cite{bouyer-timed-modal-2011}.

\begin{lemma} $\tfoAr{\phi_1}{\phi_2}$  is logically equivalent to $\tfoAsh{\phi_2} \lgcor \tfoEr{\phi_2}{\phi_1 \lgcand \phi_2}$.
\label{lem:deriv}
\end{lemma}

\begin{theorem}[Soundness and Completeness] The additional $L^{\mathit{rel}, \mathit{af}}_{\nu,\mu}$ proof rules  are sound and complete:  for any $L^{\mathit{rel}, \mathit{af}}_{\nu,\mu}$ formula $\psi$ and any state $(l,\nu)$, $(l,\nu) \in \satset{\psi}_{M,TA,\theta}$ if and only if $(l, cc_{\nu}) \vdash \psi$.
\label{thm:soundcomplete}
\end{theorem}

\section{Optimizing Performance via Derived Proof Rules}
\label{s:optderiv}

To simplify reasoning about soundness and completeness, sets of proof rules are often kept small and simple. However, we can improve the performance or proof search by having the computer work with \emph{derived} proof rules. We describe two such situations where we use derived proof rules. We discuss a third situation, invariants, in Appendix C.2.

\textbf{Optimizing $\vee$.} For performance reasons we replace a rule for $\vee$ in \cite{zhang-fast-generic-2005,zhang-fast-on-the-fly-2005}. Those papers use the proof rule $\vee_{c}$ given in Figure \ref{fig:selectpr}. We instead use the proof rule $\vee_{s}$, which we give in Figure \ref{fig:newrules}.  By pushing fresh placeholders for both branches, we avoid computing the complementation operator, which often results in forming a placeholder involving a union of clock constraints. 

\textbf{Optimizing $\tfoAr{\psi_1}{\psi_2}$.} Recall the derived formula for $\tfoAr{\psi_1}{\psi_2}$ from Lemma \ref{lem:deriv}:
$\tfoAr{\psi_1}{\psi_2}$ is equivalent to  $\tfoAsh{\psi_2} \lgcor \tfoEr{\psi_2}{\psi_1 \lgcand \psi_2}$.
This formula requires $\psi_{2}$ to be checked three times. However, by modifying the proof rule, we notice that we can perform the checking of $\psi_2$ \emph{only once}. First, we rewrite $\tfoEr{\psi_2}{\psi_1 \lgcand \psi_2}$ as $\tfoEr{\leq, \psi_2}{\psi_1}$, pushing the boundary case into the left subformula. Second, the key is to compute the largest placeholder that satisfies $\psi_2$, to remember those states (memoize), and then to reason with this placeholder (and its time predecessor) to find the placeholders needed to satisfy the two branches of the derived formula. This reasoning allows the tool to reason with the subformula $\psi_2$ only once, reusing the obtained information. The derived proof rules are in Figure \ref{fig:farpr}. The first two handle the simpler cases when either $\psi_2$ is always true (or when $\psi_1$ is always false) or $\psi_1$ is immediately true (such as when $\psi_1$ is an atomic proposition); the third rule ($\forall_{ro3}$) is the more complex case. The proof rules involving placeholders are similar. Their derivations as well as their proofs of soundness and completeness are in Appendix C.1.

\begin{figure}
\begin{proofinlinepf}
\AxiomC{\sqjudgement{\sqpremta}{\tfoAsh{\psi_{2}}}}
\RightLabel{($\forall_{ro1}$)}
\UnaryInfC{\sqjudgement{\sqpremta}{\tfoAr{\psi_1}{\psi_2}}}
\end{proofinlinepf} 
\begin{proofinlinepf}
\AxiomC{\sqjudgement{\sqpremta}{\psi_{1} \lgcand \psi_2}}
\RightLabel{($\forall_{ro2}$)}
\UnaryInfC{\sqjudgement{\sqpremta}{\tfoAr{\psi_1}{\psi_2}}}
\end{proofinlinepf} \parspc

\begin{proofinlinepf}
\AxiomC{\sqjudgement{succ(\sqpremta), \phi_{s_1}}{\psi_{1}}}
\noLine
\UnaryInfC{\sqjudgement{succ(\sqpremta), \phi_{s_2}}{\psi_{2}}}
\noLine
\UnaryInfC{\sqjudgement{succ(\sqpremta), pred(\phi_{s_1})}{succ(\sqpremta), \phi_{s_2}}}
\AxiomC{\sqjudgement{\phi_{\exists}}{pred(\phi_{s_1})}}
\noLine
\UnaryInfC{\sqjudgement{succ(\sqpremta, \phi_{\forall})}{succ(\sqpremta) \lgcand \phi_{s_2}}}
\noLine
\UnaryInfC{\sqjudgement{\sqpremta}{ \phi_{\exists} \lgcor \phi_{\forall}}}
\RightLabel{($\forall_{ro3}$)}
\BinaryInfC{\sqjudgement{\sqpremta}{\tfoAr{\psi_1}{\psi_{2}}}}
\end{proofinlinepf}
\caption{Derived proof rules for $\tfoAr{\psi_1}{\psi_2}$.}
\label{fig:farpr}
\end{figure}

\section{Implementation Details}
\label{s:mcalg}

\subsection{Addressing Non-convexity: Zone Unions}

For a subset of properties including safety properties, \emph{clock zones}, or convex sets of clock valuations, are used to make the model-checking as coarse-grained as possible. However, as shown in \cite{wang-tctl-inevitability-2006}, certain automata with certain formulas require non-convex sets of clock valuations (unions of clock zones) to be model-checked correctly. For simplicity, we use a list of Difference Bound Matrices (DBMs) to implement unions of clock zones. Other more complex data structures have been developed which include the Clock Difference Diagram (CDD) \cite{behrmann-efficient-timed-1999} and Clock Restriction Diagram (CRD) \cite{wang-efficient-verification-2004}.

\subsection{Addressing Performance: Simpler PES Formulas}
\label{ss:simplermes}

When writing safety and liveness properties, we can use the formulas from \cite{fontana-expressiveness-results-2014}. However, in the common case where there are no nested temporal operators and the formula does not involve clock constraints, we can simplify the formulations considerably. In these cases, the subformula is a conjunction and disjunction of atomic propositions, and is represented by $p$ or $q$. Here are some simplifications:
\begin{align}
\ctlAG{p} \equiv \ &Y \mesnueq p \lgcand \tfoAsh{\mmAX{Y}} \\
\ctlAF{p} \equiv \ &Y \mesmueq p \lgcor \Bigl(\tfoAsh{\mmAX{Y}} \lgcand \tfoEsh{\tfreeze{z}{\tfoAsh{z < 1}}}\Bigr) \label{eq:s2} \\
\ctlEF{p} \equiv \ &Y \mesmueq p \lgcor \tfoEsh{\mmEX{Y}} \\
\ctlEG{p} \equiv \ &Y \mesnueq p \lgcand \Bigl(\tfoEsh{\mmEX{Y}} \lgcor \tfoAsh{\tfreeze{z}{\tfoEsh{z \geq 1}}}\Bigr)
\end{align}
The correctness proofs for these simplified formulations are in Appendix C.3.

The TCTL operators here are: \ctlAG{p} (always $p$), \ctlAF{p} (inevitably $p$), \ctlEG{p} (there exists a path where always $p$), and \ctlEF{p} (possibly $p$).  
One noticeable feature is that these simplified liveness properties do not require relativization. Another noticeable feature is that the $\vee$ can be simplified to not use placeholders; consequently, $\ctlAG{p}$ and $\ctlAF{p}$ do not require placeholders. 
Additionally, our tool directly computes  $\tfoEsh{\tfreeze{z}{\tfoAsh{z < 1}}}$, time can elapse forever without an action transition, and its dual, $\tfoAsh{\tfreeze{z}{\tfoEsh{z \geq 1}}}$.

\subsection{Placeholder Implementation Complexities}
\label{ss:placecomplex}

Consider the two placeholder premises in the $\tfoAsh{\psi}$ and $\tfoEr{\psi_1}{\psi_2}$ proof rules in Figures \ref{fig:selectpr} and \ref{fig:newrules}. The placeholder sequents are given here:
\begin{equation}
\sqjudgement{succ(\sqpremta, \phi_{\forall})}{succ(\sqpremta) \lgcand \phi_{s}} \text{ and } \sqjudgement{succ(\sqpremta) ,pred_{<}(\phi_{s})}{\psi_{1}}
\end{equation}

In soundness and completeness proofs, we use soundness to give us a placeholder to show that the formula holds, and with completeness, we argue that some placeholder exists. Given the complexities of the formulas, the tool needs to find the \emph{largest} such placeholder. The rules are designed for the tool to implement them in a left-to-right fashion, where placeholders are tightened by right-hand rules. However, as the placeholders are tightened, we need to make sure that the tightened placeholder still satisfies the left-hand premise. For instance, consider the second of the above placeholders. As we tighten the placeholder to satisfy $\psi_1$, we need to check that this placeholder is the predecessor$_{<}$ of the placeholder that satisfies $\psi_2$. These checks take extra algorithmic work.

\section{Performance Evaluation}
\label{s:examples}

We present the results of an experimental evaluation of our method that demonstrates the types of timed automata and specifications the system can model check. Furthermore, on the subset of specifications that UPPAAL supports, we compare our tool's time performance to their tools's time performance.

\subsection{Methods: Evaluation Design}

In our case study, we use four different models: Carrier Sense, Multiple Access with Collision Detection (CSMA); Fischer's Mutual Exclusion (FISCHER); Generalized Railroad Crossing (GRC); and Leader election (LEADER). For more information on these models, see Appendix D.1 or \cite{heitmeyer-the-generalized-1994,zhang-fast-generic-2005,zhang-fast-on-the-fly-2005}.

For each model, we start at 4 processes and scale the model up by adding more processes (up to 8 processes). For each model we model-checked one valid safety (always) specification ($as$), one invalid safety specification ($bs$), one valid liveness (inevitably) specification ($al$), and one invalid liveness specification ($bl$). Each of these cases involves only one temporal operator: $\psi_1$ involves conjunctions and disjunctions of atomic propositions and clock constraints. In addition we tested 4 additional specifications on each property ($M1$, $M2$, $M3$, and $M4$), some of which are the leads to property $p \leadsto q$. Out of these specifications, at least one (usually $M4$) is a property with no known equivalent TCTL formula. The specifications checked are listed in Appendix D.2. The experiments were run on an Intel Mac with 8GB ram and a quad-core 2 GHz Intel Core i7 processor running OS 10.7. Times were measured with the UNIX utility \texttt{time}.

\subsection{Data and Results}

The data is provided in Tables \ref{tab:data2} and \ref{tab:data1}. Table \ref{tab:data2} contains the remaining specifications that are not supported by UPPAAL. Table \ref{tab:data1} contains the examples that are supported both by our tool (PES) and by UPPAAL (UPP), with the number indicating the number of processes used in the model.  We use the following abbreviations: TO (timeout: the example took longer than 2 hours), TOsm (the example timed out with fewer process), and O/M (out of memory). Since our tool supports a superset of the specifications that UPPAAL can support, there are specifications that our tool supports that UPPAAL does not. A scatter plot of the data in Table \ref{tab:data1} is given in Appendix D.3. 

\rowcolors{1}{}{lightgray}
\begin{table}[tb]
\caption{Examples that UPPAAL does not support. All times are in seconds (s).}
\begin{tabular}{|l|r|r|r|r|r|}
  \hline
\textbf{File} & \textbf{PES4} & \textbf{PES5} & \textbf{PES6} & \textbf{PES7} & \textbf{PES8} \\ 
  \hline
CSMA-as & 0.29 & 4.62 & 139.16  & 6696.08 & TO \\ 
   \hline
CSMA-M3 & 0.01 & 0.03 & 0.14 & 0.80 & 3.99 \\ 
   \hline
CSMA-M4 & 0.01 & 0.03 & 0.14 & 0.71 & 3.66 \\ 
   \hline
FISCHER-M3 & 0.14 & 2.51 & 79.17 & TO & TOsm \\ 
   \hline
FISCHER-M4 & 0.00 & 0.00 & 0.00 & 2.04 & 2.42 \\ 
   \hline
GRC-M2 & 0.01 & 0.01 & 0.01 & 0.02 & 0.03 \\ 
   \hline
GRC-M4 & 0.00 & 0.00 & 0.01 & 0.02 & 0.01 \\ 
   \hline
GRC-M4ap & 0.00 & 0.00 & 0.01 & 0.01 & 0.01 \\ 
   \hline
LEADER-M1 & 0.00 & 0.00 & 0.00 & 0.01 & 0.01 \\ 
   \hline
LEADER-M3 & 0.01 & 0.08 & 2.12 & 79.05 & 4242.97 \\ 
   \hline
LEADER-M4 & 0.00 & 0.00 & 0.04 & 0.03 & 0.01 \\ 
   \hline
\end{tabular}
\label{tab:data2}
\end{table}

\rowcolors{1}{}{lightgray}
\begin{table}[tb]
\caption{Time performance in seconds (s) on examples comparing PES and UPPAAL.}
\begin{tabular}{|l|rr|rr|rr|rr|rr|} 
 \hline
\textbf{File} & \textbf{PES4} & \textbf{UPP4} & \textbf{PES5} & \textbf{UPP5} & \textbf{PES6} & \textbf{UPP6} & \textbf{PES7} & \textbf{UPP7} & \textbf{PES8} & \textbf{UPP8} \\ 
 \hline
CSMA-al & 0.01 & 1.45 & 0.03 & 0.24 & 0.13 & 0.25 & 0.72 & 0.26 & 3.65 & 0.26 \\ 
   \hline
CSMA-bl & 0.01 & 0.26 & 0.03 & 0.27 & 0.13 & 0.27 & 0.73 & 0.28 & 3.53 & 0.33 \\ 
   \hline
CSMA-bs & 0.01 & 0.33 & 0.05 & 0.27 & 0.22 & 0.27 & 1.14 & 1.33 & 5.09 & 4.66 \\ 
   \hline
CSMA-M1 & 0.01 & 0.29 & 0.03 & 0.27 & 0.14 & 0.28 & 0.73 & 0.27 & 3.69 & 0.27 \\ 
   \hline
CSMA-M2 & 0.33 & 0.35 & 5.21 & 7.00 & 154.56 & 1194.74 & TO & TO  & TOsm & TOsm  \\ 
   \hline
FISCHER-al & 0.00 & 0.51 & 0.00 & 0.27 & 0.00 & 0.28 & 0.00 & 0.40 & 0.00 & 0.27 \\ 
   \hline
FISCHER-as & 0.07 & 0.27 & 0.51 & 0.28 & 13.44 & 0.67 & 864.04 & 0.96 & TO  & 4.26 \\ 
   \hline
FISCHER-bl & 0.00 & 0.26 & 0.00 & 0.26 & 0.00 & 0.28 & 0.00 & 0.34 & 0.00 & 0.26 \\ 
   \hline
FISCHER-bs & 0.04 & 0.28 & 0.01 & 0.27 & 0.02 & 0.32 & 0.39 & 0.47 & 0.39 & 0.90 \\ 
   \hline
FISCHER-M1 & 0.00 & 0.26 & 0.00 & 0.26 & 0.00 & 0.28 & 0.00 & 0.28 & 0.00 & 0.25 \\ 
   \hline
FISCHER-M2 & 0.00 & 0.26 & 0.00 & 0.26 & 0.00 & 0.27 & 0.00 & 0.30 & 0.03 & 0.28 \\ 
   \hline
GRC-al & 0.00 & 0.27 & 0.01 & 0.28 & 0.47 & 0.59 & 0.07 & 0.44 & 0.08 & 5.45 \\ 
   \hline
GRC-as & 53.09 & 0.36 & TO & 7.11 & TOsm  & 940.51 & TOsm  & 3433.14 & TOsm & TO  \\ 
   \hline
GRC-bl & 0.00 & 0.27 & 0.00 & 0.27 & 0.01 & 0.27 & 0.01 & 0.61 & 0.01 & 0.66 \\ 
   \hline
GRC-bs & 0.11 & 0.41 & 1.91 & 0.41 & 433.59 & 1.76 & O/M & 16.19 & O/M & 52.03 \\ 
   \hline
GRC-M1 & 0.01 & 0.27 & 0.04 & 0.27 & 0.01 & 0.29 & 0.05 & 0.35 & 0.03 & 0.32 \\ 
   \hline
GRC-M3 & 0.00 & 0.27 & 0.00 & 0.31 & 0.01 & 0.56 & 0.04 & 1.23 & 0.01 & 3.85 \\ 
   \hline
LEADER-al & 0.00 & 0.28 & 0.01 & 0.33 & 0.17 & 4.30 & 5.80 & 747.82 & 573.84 & TO  \\ 
   \hline
LEADER-as & 0.00 & 0.27 & 0.01 & 0.27 & 0.22 & 0.33 & 6.23 & 0.86 & 649.52 & 8.21 \\ 
   \hline
LEADER-bl & 0.00 & 0.28 & 0.00 & 0.27 & 0.01 & 0.28 & 0.17 & 0.32 & 4.25 & 0.29 \\ 
   \hline
LEADER-bs & 0.00 & 0.27 & 0.00 & 0.28 & 0.01 & 0.28 & 0.03 & 4.99 & 0.40 & 1.57 \\ 
   \hline
LEADER-M2 & 0.00 & 0.28 & 0.02 & 0.31 & 0.38 & 3.05 & 13.53 & 504.89 & 1570.37 & TO  \\ 
   \hline 
\end{tabular}
\label{tab:data1}
\end{table}

\subsection{Analysis and Discussion}

After analyzing the data, we may draw three conclusions. First, on the examples that both our PES tool and UPPAAL support, we see that UPPAAL's performance is generally faster than ours, although, our tool performs faster on some examples. Additionally, while our tool does time out more often than UPPAAL does, most examples are verified quickly by both tools.
	Second, our tool can reasonably efficiently verify specifications that UPPAAL cannot. 
	Third, for these examples, the performance bottleneck seems to be safety properties. Even with the additional complexity of supporting the more complicated specifications (in both tables), liveness was often verified more quickly than safety properties. Here is one possible explanation: while the verifier must check the entire state space for a valid safety property, often only a subset of the state space must be checked for a liveness property. 

\section{Conclusion}
\label{s:conclusion}

	We provide the first implementation of a $L^{\mathit{rel}, \mathit{af}}_{\nu,\mu}$ timed automata model checker. Additionally, this model checker is on-the-fly, allowing for verification to explore both the timed automaton and the $L^{\mathit{rel}, \mathit{af}}_{\nu,\mu}$ formula incrementally. To support the full fragment of this logic, we extended the proof-rule framework of \cite{zhang-fast-generic-2005,zhang-fast-on-the-fly-2005} to support the relativization operators, and we optimize the tool's performance using derived proof rules. We also provided simpler $L^{\mathit{rel}, \mathit{af}}_{\nu,\mu}$ formulas for common safety and liveness formulas. While these may seem to be straightforward extensions, the rules and the extensions were \emph{designed} to be straightforward, designing the proof rules to be both easy to implement efficiently. 
	
	We then compared our tool to UPPAAL. While UPPAAL seems to perform faster more often, our tool is competitive for many of those examples, including liveness formulas.  Additionally, our tool was able to quickly verify specifications that UPPAAL does not currently support.
	
	Future work is to both further optimize the performance of our tool and to augment our tool to provide more information than just a yes or no answer. Potential information includes providing answers to these questions: Was the formula true because the premise of an implication was always false? Was the formula true because certain states were never reached? 

\section{Acknowledgements}

We thank Dezhuang Zhang for providing the code base \cite{zhang-fast-generic-2005} and for his insights.

\bibliographystyle{splncs03}
\bibliography{FORMATS2014PfRc}

\newpage

%
%
\appendix
\label{s:appendix}

\section{Timed Automata and $L^{af}_{\nu,\mu}$ Proof Rules}
\label{s:a:proofrules}

We take the proof-rule framework of \cite{zhang-fast-generic-2005,zhang-fast-on-the-fly-2005} and adapt it for timed automata and $L^{af}_{\nu,\mu}$ MES. The complete set of proof rules is given in Figures \ref{fig:allrules1} and \ref{fig:allrules2}. Figure  \ref{fig:allrules1} contains the adaptations of the rules without placeholders, and Figure \ref{fig:allrules2} contains the adaptations of the rules involving placeholders. Conditions on the proof rules are given after the rule label; the rule labels are in (), and the conditions are outside parentheses.

\begin{figure}
\begin{proofinlinepf}
\AxiomC{Premise 1}
\AxiomC{$\ldots$}
\AxiomC{Premise $n$}
\RightLabel{(\emph{Rule Template})}
\TrinaryInfC{Conclusion}
\end{proofinlinepf} 
\begin{proofinlinepf}
\AxiomC{}
\RightLabel{($p$), $p \in M(l)$}
\UnaryInfC{\sqjudgement{\sqpremta}{p}}
\end{proofinlinepf}\parspc

\begin{proofinlinepf}
\AxiomC{}
\RightLabel{($cc'$), $cc \models cc'$}
\UnaryInfC{\sqjudgement{\sqpremta}{cc'}}
\end{proofinlinepf}
\begin{proofinlinepf}
\AxiomC{\sqjudgement{\sqpremta}{\psi_i}}
\RightLabel{($p$), $X_i \overset{\nu/\mu}{=} \psi_i$}
\UnaryInfC{\sqjudgement{\sqpremta}{X_i}}
\end{proofinlinepf}
\begin{proofinlinepf}
\AxiomC{}
\RightLabel{(Empty)}
\UnaryInfC{\sqjudgement{(l,\lgfalse)}{\psi}}
\end{proofinlinepf}\parspc

\begin{proofinlinepf}
\AxiomC{$\sqpremta \vdash \psi_1$}
\RightLabel{($\vee_l$)}
\UnaryInfC{$\sqpremta \vdash \psi_1 \lgcor \psi_2$}
\end{proofinlinepf}
\begin{proofinlinepf}
\AxiomC{$\sqpremta \vdash \psi_2$}
\RightLabel{($\vee_r$)}
\UnaryInfC{$\sqpremta \vdash \psi_1 \lgcor \psi_2$}
\end{proofinlinepf} 
\begin{proofinlinepf}
\AxiomC{$\sqpremta \vdash \psi_1$}
\AxiomC{$\sqpremta \vdash \psi_2$}
\RightLabel{($\wedge$)}
\BinaryInfC{$\sqpremta \vdash \psi_1 \lgcand \psi_2$}
\end{proofinlinepf}\parspc

\begin{proofinlinepf}
\AxiomC{\sqjudgement{\sqpremta, \phi_s}{\psi_1}}
\AxiomC{\sqjudgement{\sqpremta, \lgcnot \phi_{s}}{\psi_2}}
\RightLabel{($\vee_{c}$)}
\BinaryInfC{\sqjudgement{\sqpremta}{\psi_1 \lgcor \psi_2}}
\end{proofinlinepf} 
\begin{proofinlinepf}
\AxiomC{\sqjudgement{succ(\sqpremta)}{\psi}}
\RightLabel{($\forall_{t1}$)}
\UnaryInfC{\sqjudgement{\sqpremta}{\tfoAsh{\psi}}}
\end{proofinlinepf} \parspc

\begin{proofinlinepf}
\AxiomC{\sqjudgement{succ(\sqpremta), \phi_{s}}{\psi}}
\AxiomC{\sqjudgement{\sqpremta}{pred(\phi_{s})}}
\RightLabel{($\exists_{t1}$)}
\BinaryInfC{\sqjudgement{\sqpremta}{\tfoEsh{\psi}}}
\end{proofinlinepf}  \parspc

\begin{proofinlinepf}
\AxiomC{\sqjudgement{(l, post(cc, \lambda := 0))}{\psi}}
\RightLabel{($[]$)}
\UnaryInfC{\sqjudgement{\sqpremta}{\psi[\lambda := 0]}}
\end{proofinlinepf} \parspc

\begin{proofinlinepf}
\AxiomC{\sqjudgement{\sqpremta}{\mmAXra{a_1}{\psi}}}
\AxiomC{\ldots}
\AxiomC{\sqjudgement{(\sqpremta}{\mmAXra{a_n}{\psi}}}
\RightLabel{($\mmAXop_{Act}$), $\Sigma = \mset{a_1, \ldots, a_n}$}
\TrinaryInfC{\sqjudgement{\sqpremta}{\mmAX{\psi}}}
\end{proofinlinepf} \parspc

\begin{proofinlinepf}
\AxiomC{\sqjudgement{(l_n,cc \lgcand g)}{\psi[\lambda := 0]}}
\RightLabel{($\mmEXraop{a}_{Act}$), $(l,a,g,\lambda,l') \in E$, $cc \lgcand g$ is satisfiable}
\UnaryInfC{\sqjudgement{\sqpremta}{\mmEXra{a}{\psi}}}
\end{proofinlinepf} \parspc

\begin{proofinlinepf}
\AxiomC{\sqjudgement{(\sqpremta}{\mmEXra{a_i}{\psi}}}
\RightLabel{($\mmEXop_{Act}$), $a_i \in \Sigma$}
\UnaryInfC{\sqjudgement{\sqpremta}{\mmEX{\psi}}}
\end{proofinlinepf} \parspc

\caption{Proof rules (without placeholders) adapted for timed automata and MES.}
\label{fig:allrules1}
\end{figure}

Most of the rules involving placeholders are similar (the placeholder is just an additional premise). This next figure contains the remaining rules.

\begin{figure}
\begin{proofinlinepf}
\AxiomC{Premise 1}
\AxiomC{$\ldots$}
\AxiomC{Premise $n$}
\RightLabel{(\emph{Rule Template})}
\TrinaryInfC{Conclusion}
\end{proofinlinepf} \parspc

\begin{proofinlinepf}
\AxiomC{\sqjudgement{succ(\sqpremta), \phi_{s}}{\psi}}
\AxiomC{\sqjudgement{succ(\sqpremta, \phi_{\forall})}{succ(\sqpremta) \lgcand \phi_{s}}}
\RightLabel{($\forall_{t2}$)}
\BinaryInfC{\sqjudgement{\sqpremta, \phi_{\forall}}{\tfoAsh{\psi}}}
\end{proofinlinepf} \parspc

\begin{proofinlinepf}
\AxiomC{\sqjudgement{succ(\sqpremta), \phi_{s}}{\psi}}
\AxiomC{\sqjudgement{\phi_{\exists}}{pred(\phi_{s})}}
\RightLabel{($\exists_{t2}$)}
\BinaryInfC{\sqjudgement{\sqpremta, \phi_{\exists}}{\tfoEsh{\psi}}}
\end{proofinlinepf}  \parspc

\begin{proofinlinepf}
\AxiomC{\sqjudgement{(l, post(cc, \lambda := 0)), \phi_{s}}{\psi}}
\AxiomC{\sqjudgement{\phi_{[]}}{\phi_{s}[\lambda := 0]}}
\RightLabel{($[]_p$)}
\BinaryInfC{\sqjudgement{\sqpremta, \phi_{[]}}{\psi[\lambda := 0]}}
\end{proofinlinepf}

\caption{Proof rules (involving placeholders) adapted for timed automata and MES.}
\label{fig:allrules2}
\end{figure}

A note about clock resets (substations of clocks to $0$ in $\psi$). The formal definition uses the $post$ operator from \cite{zhang-fast-generic-2005,zhang-fast-on-the-fly-2005}, defined as:
\begin{equation}
post(cc, \lambda := e) \overset{def}{=} \foEo{v}{\bigl(\lambda = (e[\lambda := v]) \lgcand cc[\lambda := v]\bigr)} 
\end{equation}
In the special case of resetting clocks to $0$, computing post results in one of two cases:
\begin{enumerate}
	\item If the original clock constraint $cc$ is unsatisfiable, $post((l,cc), \lambda := 0)$ produces $(l,cc')$ where $cc'$ is logically equivalent to $\lgfalse$.
	\item Otherwise, $cc$ is a satisfiable clock constraint, and $post((l,cc), \lambda := 0)$ becomes $(l,reset(cc, \lambda := 0)$, where $reset(cc, \lambda :=0)$ is the clock zone reset operator given in \cite{alur-timed-automata-1999}.
\end{enumerate}

\section{Proofs of Results}
\label{s:a:proofs}

When discussing the proofs of the formulas, we will use \emph{extended states.}  To handle freeze quantification, timed automaton states $(l,\nu)$ are extended to handle freeze quantification ($\tfreeze{z}{\psi}$); the new \emph{extended state} is $(l,\nu,\nu_{f})$ and $\nmfunc{\nu_{f}}{CX_{f}}{\Rr^{\geq 0}}$ is the valuation for all the clocks introduced by freeze quantification. This formalism comes from \cite{bouyer-timed-modal-2011}. When clear from context, we will refer to an extended state as $(l,\nu)$ and omit the explicit notation of $\nu_{f}$. 

\begin{proof}[Proof of Lemma \ref{lem:deriv}] We prove both directions of the lemma.  Let $TA$ be an arbitrary timed automaton and let $(l,\nu, \nu_{f})$ be an arbitrary (extended) state in $TA$. 

First, suppose that $(l,\nu, \nu_{f}) \models \tfoAr{\psi_1}{\psi_2}$. By definition, $\forall \delta \geq 0: \bigl((l, \nu + \delta, \nu_{f} + \delta) \models \psi_2 
\lgcor \exists \delta', 0 \leq \delta' < \delta: (l, \nu + \delta', \nu_{f} + \delta') \models \psi_1 \bigr)$. Notice that the entire quantification is inside the $\forall \delta$. Now for each time instance, one of the two disjunctions inside the $\forall$ is true. We split the cases into two cases: \parspc

\noindent \textbf{Case 1:} The $\psi_2$ disjunct is always true. 

Then the formula reduces to $\forall \delta \geq 0: \bigl((l, \nu + \delta, \nu_{f} + \delta) \models \psi_2$. By definition, this means that $(l,\nu,\nu_{f}) \models \tfoAsh{\psi_2}$. \parspc

\noindent \textbf{Case 2:} The $\exists$ disjunct is satisfied for at least one such $\delta$.

While $\Rr^{\geq 0}$ is not well-ordered with respect to $\leq$, we utilize that there are a finite number of constraints involving a finite number of (possibly not integer) constants. We can then group all $\delta$ values into groups of consecutive values based on whether the group of $\delta$ values satisfy $\psi_1$, $\psi_2$, both, or neither. By construction of $\psi_1$ and $\psi_2$, the finest-grained groups are clock regions (the same regions used in region equivalence) and each time advance $\delta$ appears in exactly one group. Hence, we have a finite number of groups of $\delta$, which are well ordered. (We can well order the groups by the largest $\delta$ in each group. Also note that $\delta \in \Rr^{\geq 0}$.)  

Let $\delta_{s}$ be a time advance in the smallest such group. Since the disjunct is satisfied at time $\delta_{s}$, there is some time $\delta_{p} < \delta_{s}, (l, \nu + \delta_{p}, \nu_{f} + \delta_{p}) \models \psi_1$. Let $\delta_{p}$ be the time when $\psi_1$ is satisfied (there may be many of these, but for this definition, we do not care.) Likewise, since $\delta_{s}$ is in the smallest such $\delta$ group, we know that all smaller $\delta$ (smaller than the group of $\delta$ values that $\delta_{s}$ is in) satisfy the left disjunct, meaning that $\tfoA{\delta'' < \delta_{s}}{(l, \nu + \delta'', \nu_{f} + \delta'') \models \psi_2}$. Since $\delta_{p} < \delta_{s}$ (and $\delta_{p}$ is smaller than any delta in the group of $\delta_{s}$), we have $(l,\nu + \delta_{p},\nu_{f} + \delta_{p}) \models \psi_1 \lgcand \psi_2 \lgcand \tfoA{\delta' < \delta_{p}}{(l,\nu + \delta', \nu_{f} + \delta') \models \psi_2}$. Using $\delta_{p}$ as the chosen delta for $\exists$, we have that $(l,\nu,\nu_{f}) \models \tfoEr{\psi_2}{\psi_1 \lgcand \psi_2}$. \parspc

Now we suppose that $(l,\nu,\nu_{f}) \models \tfoA{\psi_2} \lgcor \tfoEr{\psi_2}{\psi_1 \lgcand \psi_2}$. This direction is similar to the previous direction. We break the case on which disjunct is satisfied. \parspc

\noindent \textbf{Case 1:} $(l,\nu,\nu_{f}) \models \tfoAsh{\psi_2}$.

By definition, $(l,\nu,\nu_{f}) \models \tfoAr{\psi_1}{\psi_2}$. (The above case is the special case with $\psi_1 = \lgfalse$, which is harder to satisfy.) \parspc

\noindent \textbf{Case 2:} $(l,\nu,\nu_{f}) \models \tfoEr{\psi_2}{\psi_1 \lgcand \psi_2}$.

Let $\delta_{e}$ be the chosen time when $\psi_1 \lgcand \psi_2$ is true. Now we handle all time advances $\delta$. For all time advances $\delta > \delta_{e}$. Since $\psi_1$ is true, $\tfoAr{\psi_1}{\psi_2}$ is satisfied for those times. For all times $\delta < \delta_{e}$, since $\psi_2$ is true, $\tfoAr{\psi_1}{\psi_2}$ is satisfied for those times. When $\delta = \delta_{e}$, we need $\psi_2$ to be true, which it is.

Hence, for all time advances $\delta$, the definition of $\tfoAr{\psi_1}{\psi_2}$ is satisfied. Hence $(l,\nu,\nu_{f}) \models \tfoAr{\psi_1}{\psi_2}$.\qed
\end{proof}

\begin{proof}[Proof of Theorem \ref{thm:soundcomplete}] 
We use the soundness and completeness proofs of rules  in \cite{zhang-fast-generic-2005,zhang-fast-on-the-fly-2005}. Hence, we only need to prove the correctness of the proof rules we provided in this paper.We now prove the soundness and completeness of the $\tfoEr{\psi_1}{\psi_2}$  and the $\vee_{s}$. proof rules.

First, we start with the proof rule for $\vee_{s}$.

\textbf{Soundness:} Suppose this rule is true. Then $\psi_{s_1}$ acts as the placeholder $\psi_{s}$. Given that $z_{\infty} \subseteq (\psi_{s_1} \vee \psi_{s_2})$, we know that $\lgnot{\psi_s} \subseteq \psi_{s_2}$, since $\lgnot{\psi_{s}} = z_{\infty} - \psi_s$ by definition of complement.

\textbf{Completeness:} Suppose the conclusion is indeed true. Then by the completeness of $\vee_{s}$, we can use that rule. Choose $\psi_{s_1} = \psi_{s}$ and $\psi_{s_2} = \lgnot{\psi_{s}}$ By definition of $\lgnot{}$, $z_{\infty} \subseteq (\psi_{s_1} \lgcor \psi_{s_2}) = (\psi_{s} \lgcor \lgnot{\psi_{s}}) = z_{\infty}$. \parspc

Now we prove the correctness of the remaining $\tfoEr{\psi_1}{\psi_2}$ proof rules. We start with Rule $\exists_{r1}$.

\textbf{Soundness:} Suppose this rule is true. By the correctness of $\exists_{t_1}$ in \cite{zhang-fast-generic-2005,zhang-fast-on-the-fly-2005}, we know that $\Gamma \vdash \tfoEsh{\psi_2}$. We now need to argue that $\psi_1$ is true until $\psi_2$ is true. Examine the valuation set $succ(\Gamma) \cap \phi_s$. By construction, $succ(\Gamma)$ gives all possible valuations after some time advance from $\Gamma$. By the constraint $\Gamma \vdash pred(\phi_s)$, we know that $\Gamma$ must be able to time-lapse to each valuation in the clock zone $\phi_s$. Hence, $succ(\Gamma) \cap pred(\phi_s)$ is the set of valuations that elapses to $succ(\Gamma) \cap \phi_{s}$, and $succ(\Gamma) \cap pred_{<}(\phi_{s})$ is that set of valuations requiring some non-zero time elapse to $\phi_s$. From the truth of this second premise, $\psi_1$ is must be true for all of those times. 

\textbf{Completeness:} Suppose the conclusion is indeed true; suppose $\Gamma \vdash \tfoEr{\psi_1}{\psi_2}$.  Therefore, there is a time advance $t$ such that after elapsing $t$ units, $\psi_2$ is true, and for all times until (and not including $t$), $t$, $\psi_1$ is true. Choose $\phi_{s}$ to be a placeholder such that $succ(\Gamma) \wedge \phi_{s}$ is the sequent $\Gamma$ after $t$ units has elapsed. By the correctness of the proof for the rule $\exists_{t1}$ in \cite{zhang-fast-generic-2005,zhang-fast-on-the-fly-2005}, we know that $succ(\Gamma), \phi_{s} \vdash \psi_2$. Because for all times until the times $\phi_s$, $\psi_{1}$ is true, this set of times by definition is $succ(\Gamma) \wedge pred_{<}(\phi_s)$. Therefore, $succ(\Gamma) \wedge pred_{<}(\phi_s) \vdash \psi_{1}$ By definition of the time elapse, since this is true we know that $\Gamma \vdash pred(\phi_{s})$. \parspc

We now examine rule $\exists_{r2}$. Its proof of soundness and completeness is similar to $\exists_{r1}$. The difference is that we are elapsing from $\Gamma \cap \phi_{s}$. In $\exists_{t2}$, the third clause shrinks $\phi_{s}$ to ensure that the time-elapse relation holds.

\textbf{Soundness:} Suppose this rule is true. By the correctness of $\exists_{t2}$, the first and third premise show that $\tfoEsh{\psi_2}$ is true. Let that time advance be $\delta$ units.  We now need to argue that $\psi_1$ is true until $\psi_2$ is true. Examine the valuation set $succ(\Gamma) \cap \phi_{s'}$. By construction, $succ(\Gamma)$ gives all possible valuations after some time advance from $\Gamma$. By the constraint $\phi_{s} \vdash pred(\phi_{s'})$, we know that $\Gamma, \phi_{s}$ must be able to time-lapse to each valuation in the clock zone $\phi_{s'}$. Hence, $succ(\Gamma) \cap pred(\phi_{s'})$ is the set of valuations that elapses to $succ(\Gamma) \cap \phi_{s'}$, and $succ(\Gamma) \cap pred_{<}(\phi_{s'})$ is that set of valuations requiring some non-zero time elapse to $\phi_{s'}$. From the truth of this second premise, $\psi_1$ is must be true for all of those times. 

\textbf{Completeness:} Suppose the conclusion is indeed true; suppose $\Gamma, \phi_{s} \vdash \tfoEr{\psi_1}{\psi_2}$.  Therefore, there is a time advance $t$ such that after elapsing $t$ units, $\psi_2$ is true, and for all times until (and not including $t$) $t$, $\psi_1$ is true. Choose $\phi_{s'}$ to be a placeholder such that $succ(\Gamma) \wedge \phi_{s'}$ is the sequent $\Gamma, \phi_{s}$ after $t$ units has elapsed. By the correctness of the proof for the rule $\exists_{t2}$ in \cite{zhang-fast-generic-2005,zhang-fast-on-the-fly-2005}, we know that $succ(\Gamma), \phi_{s'} \vdash \psi_2$. Because for all times until the times $\phi_{s'}$, $\psi_{1}$ is true, this set of times by definition is $succ(\Gamma) \wedge pred_{<}(\phi_{s'})$. Therefore, $succ(\Gamma) \wedge pred_{<}(\phi_{s'}) \vdash \psi_{1}$ By definition of the time elapse, since this is true we know that $\phi_{s} \vdash pred(\phi_{s'})$. \qed
\end{proof}

\section{Derived Proof Rules}
\label{s:a:derivedrules}

\subsection{Derived $\tfoAr{\psi_1}{\psi_2}$ Rules}
\label{ss:a:forallderiv}

To derive an optimized proof for $\tfoAr{\psi_1}{\psi_2}$, we first will derive a slightly different version for $\tfoEr{\psi_2}{\psi_1 \lgcand \psi_2}$, rewriting the proof rule for this special case slightly. This version will allow us to get the same premise to appear twice in the proof.

The $\tfoEr{\psi_{1}}{\psi_{1} \lgcand \psi_{2}}$ can be slightly rewritten as follows:

\begin{prooftreepf}
\AxiomC{\framebox{\sqjudgement{succ(\sqpremta), \phi_{s_1}}{\psi_{1}}}}
\AxiomC{\framebox{\sqjudgement{succ(\sqpremta), pred(\phi_{s_1})}{\psi_{2}}}}
\AxiomC{\framebox{\sqjudgement{\phi_{\exists}}{pred(\phi_{s_1})}}}
\RightLabel{($\exists_{r2}$ rewrite)}
\TrinaryInfC{\sqjudgement{\sqpremta, \phi_{\exists}}{\tfoEr{\psi_{2}}{\psi_{1} \lgcand \psi_{2}}}}
\UnaryInfC{$\exists$ ph rewrite}
\end{prooftreepf}

Hence, we look more closely at the rewrite rule, comparing the derivation that is obtained, we get the following proof rule, $\exists_{rw}$:
\begin{prooftreepf}
\AxiomC{\sqjudgement{succ(\sqpremta), \phi_{s_1}}{\psi_{1}};
\sqjudgement{succ(\sqpremta), pred(\phi_{s_1})}{\psi_{2}}; \sqjudgement{\phi_{\exists}}{pred(\phi_{s_1})}}
\RightLabel{($\exists_{rw}$)}
\UnaryInfC{\sqjudgement{succ(\sqpremta), \phi_{s_1}}{\psi_{1} \lgcand \psi_{2}};
\sqjudgement{succ(\sqpremta), pred_{<}(\phi_{s_1})}{\psi_{2}};
\sqjudgement{\phi_{\exists}}{pred(\phi_{s_1})}}
\end{prooftreepf}

This rule has three conclusions that are written as three premises. 

\begin{lemma}
The $\exists_{rw}$ rule is sound and complete.
\label{lem:derivpr1}
\end{lemma}

\begin{proof}[Proof of Lemma \ref{lem:derivpr1}]

\textbf{Soundness:} Assume that the three top premises are true. Since the third conclusion is the same as the third premise, that is true. Now we must argue that $\sqjudgement{succ(\sqpremta), \phi_{s_1}}{\psi_{1} \lgcand \psi_{2}}$ and $\sqjudgement{succ(\sqpremta), pred_{<}(\phi_{s_1})}{\psi_{2}}$. Since we have $\sqjudgement{succ(\sqpremta), pred(\phi_{s_1})}{\psi_{s_2}}$ and $pred_{<}(phi_{s_1}) \subseteq pred_{<}(phi_{s_1})$, we know $\sqjudgement{succ(\sqpremta), pred(\phi_{s_1})}{\psi_{s_2}}$. Furthermore, since  $\phi_{s_1} \subseteq pred(\phi_{s_1})$, $\sqjudgement{succ(\sqpremta), \phi_{s_1}}{\psi_{2}}$. Therefore, $\sqjudgement{succ(\sqpremta), \phi_{s_1}}{\psi_{1} \lgcand \psi_{2}}$.

\textbf{Completeness:} Assume that the three bottom conclusions are true. Since the third premise is the same as the third conclusion, that is true. Now we must argue that $\sqjudgement{succ(\sqpremta), \phi_{s_1}}{\psi_{1}}$ and $\sqjudgement{succ(\sqpremta), pred(\phi_{s_1})}{\psi_{2}}$. Since we have $\sqjudgement{succ(\sqpremta), \phi_{s_1}}{\psi_1 \lgcand \psi_2}$, we have $\sqjudgement{succ(\sqpremta), \phi_{s_1}}{\psi_1}$. Furthermore, since $\phi_{s_1} \cup pred_{<}(\phi_{s_1}) = pred(\phi_{s_1})$, we know that $\sqjudgement{succ(\sqpremta), pred(\phi_{s_1})}{\psi_{s_2}}$. \qed
\end{proof}

Now, with the rule $\exists_{rw}$, we utilize the formulation in Lemma \ref{lem:deriv} to derive a proof for $\tfoAr{\psi_1}{\psi_2}$. Here is the derivation for $\tfoAr{\psi_1}{\psi_2}$.

\begin{prooftreepf}
\AxiomC{\framebox{\sqjudgement{succ(\sqpremta), \phi_{s_1}}{\psi_{1}}}}
\AxiomC{\framebox{\sqjudgement{succ(\sqpremta), pred(\phi_{s_1})}{\psi_{2}}}}
\AxiomC{\framebox{\sqjudgement{\phi_{\exists}}{pred(\phi_{s_1})}}}
\RightLabel{($\exists_{rw}$)}
\TrinaryInfC{\sqjudgement{succ(\sqpremta), \phi_{s_1}}{\psi_{1} \lgcand \psi_{2}};
\sqjudgement{succ(\sqpremta), pred_{<}(\phi_{s_1})}{\psi_{2}};
\sqjudgement{\phi_{\exists}}{pred(\phi_{s_1})}}
\RightLabel{($\exists_{r2}$)}
\UnaryInfC{\sqjudgement{\sqpremta, \phi_{\exists}}{\tfoEr{\psi_{2}}{\psi_{1} \lgcand \psi_{2}}}}
\UnaryInfC{$\exists$ ph}
\end{prooftreepf}

\begin{prooftreepf}
\AxiomC{\framebox{\sqjudgement{succ(\sqpremta), \phi_{s_2}}{\psi_{2}}}}
\AxiomC{\framebox{\sqjudgement{succ(\sqpremta, \phi_{\forall})}{succ(\sqpremta) \lgcand \phi_{s_2}}}}
\RightLabel{($\forall_{t2}$)}
\BinaryInfC{\sqjudgement{\sqpremta, \phi_{\forall}}{\tfoAsh{\psi_{2}}}}
\UnaryInfC{$\forall$ ph}
\end{prooftreepf}

\begin{prooftreepf}
\AxiomC{See $\exists$ ph}
\UnaryInfC{\sqjudgement{\sqpremta, \phi_{\exists}}{\tfoEr{\psi_{2}}{\psi_{1} \lgcand \psi_{2}}}}
\AxiomC{See $\forall$ ph}
\UnaryInfC{\sqjudgement{\sqpremta, \phi_{\forall}}{\tfoAsh{\psi_{2}}}}
\AxiomC{\framebox{\sqjudgement{\sqpremta}{ \phi_{\exists} \lgcor \phi_{\forall}}}}
\RightLabel{($\vee_{s_r}$)}
\TrinaryInfC{\sqjudgement{\sqpremta}{\tfoEr{\psi_2}{\psi_{1} \lgcand \psi_{2}} \lgcor \tfoAsh{\psi_{2}}}}
\RightLabel{Lemma 5.1}
\UnaryInfC{\sqjudgement{\sqpremta}{\tfoAr{\psi_1}{\psi_{2}}}}
\end{prooftreepf}

We stop the derivation here and examine the derived sequents. Notice that we are computing a placeholder for $\psi_{2}$ twice. We can perform this computation \textbf{once} and save ourselves a good amount of computation time. We can compute things in this order:

\begin{enumerate}
	\item Find the placeholder for $\psi_{1}$. Utilize simpler proof rules if $\psi_{1}$ is one of the easier cases.
	\item Find the placeholder for $\psi_{2}$. (Copy this value for the other branch.) Now using this, solve the $\forall_{t2}$ rule to obtain a placeholder $\phi_{\forall}$. 
	\item After solving $\phi_{\forall}$, use the solved $\psi_{2}$ placeholder to solve for $\phi_{\exists}$. (This is the hard step that yields the optimization.)
	\item Now solve the $\lgcor_{s}$ placeholder rule.
\end{enumerate}

Using this insight, we now give the optimized proof rules. These rules are:

\begin{prooftreepf}
\AxiomC{\sqjudgement{\sqpremta}{\tfoAsh{\psi_{2}}}}
\RightLabel{($\forall_{ro1}$)}
\UnaryInfC{\sqjudgement{\sqpremta}{\tfoAr{\psi_1}{\psi_2}}}
\end{prooftreepf} 

\begin{prooftreepf}
\AxiomC{\sqjudgement{\sqpremta}{\psi_{1} \lgcand \psi_2}}
\RightLabel{($\forall_{ro2}$)}
\UnaryInfC{\sqjudgement{\sqpremta}{\tfoAr{\psi_1}{\psi_2}}}
\end{prooftreepf}

\begin{prooftreepf}
\AxiomC{\sqjudgement{succ(\sqpremta), \phi_{s_1}}{\psi_{1}}}
\noLine
\UnaryInfC{\sqjudgement{succ(\sqpremta), \phi_{s_2}}{\psi_{2}}}
\noLine
\UnaryInfC{\sqjudgement{succ(\sqpremta), pred(\phi_{s_1})}{succ(\sqpremta), \phi_{s_2}}}
\AxiomC{\sqjudgement{\phi_{\exists}}{pred(\phi_{s_1})}}
\noLine
\UnaryInfC{\sqjudgement{succ(\sqpremta, \phi_{\forall})}{succ(\sqpremta) \lgcand \phi_{s_2}}}
\noLine
\UnaryInfC{\sqjudgement{\sqpremta}{ \phi_{\exists} \lgcor \phi_{\forall}}}
\RightLabel{($\forall_{ro3}$)}
\BinaryInfC{\sqjudgement{\sqpremta}{\tfoAr{\psi_1}{\psi_{2}}}}
\end{prooftreepf}

The first proof rule is used when $\psi_{1}$ is never satisfied for any time advance. Such as case is when $\psi_{1}$ is a false atomic proposition. The second proof rule is when $\psi_{1}$ is immediately true without a time advance, such as a true atomic proposition. The third proof rule uses the expanded derived proof rule to enforce that $\psi_{2}$ is only computed once. The premise $pred(\phi_{s_1}) \subseteq \phi_{s_2}$ ensures that all of the predecessor of $\phi_1$ satisfies $\psi_{2}$. The last premise of ($\forall$ ph1) is mostly solved by checking that $pred(\phi_{s_1}) \subseteq \phi_{s_2}$; we require the intersection with the successor for completeness, since we need not require times before $\sqpremta$ to satisfy $\psi_{2}$.

If a placeholder is involved, we use the following analogous optimized proof rules:

\begin{prooftreepf}
\AxiomC{\sqjudgement{\sqpremta, \phi_{s}}{\tfoAsh{\psi_{2}}}}
\RightLabel{($\forall_{rop1}$)}
\UnaryInfC{\sqjudgement{\sqpremta, \phi_{s}}{\tfoAr{\psi_1}{\psi_2}}}
\end{prooftreepf} 

\begin{prooftreepf}
\AxiomC{\sqjudgement{\sqpremta, \phi_{s}}{\psi_{1} \lgcand \psi_2}}
\RightLabel{($\forall_{rop2}$)}
\UnaryInfC{\sqjudgement{\sqpremta, \phi_{s}}{\tfoAr{\psi_1}{\psi_2}}}
\end{prooftreepf} 

\begin{prooftreepf}
\AxiomC{\sqjudgement{succ(\sqpremta), \phi_{s_1}}{\psi_{1}}}
\noLine
\UnaryInfC{\sqjudgement{succ(\sqpremta), \phi_{s_2}}{\psi_{2}}}
\noLine
\UnaryInfC{\sqjudgement{succ(\sqpremta), pred(\phi_{s_1})}{succ(\sqpremta), \phi_{s_2}}}
\AxiomC{\sqjudgement{\phi_{\exists}}{pred(\phi_{s_1})}}
\noLine
\UnaryInfC{\sqjudgement{succ(\sqpremta, \phi_{\forall})}{succ(\sqpremta) \lgcand \phi_{s_2}}}
\noLine
\UnaryInfC{\sqjudgement{\sqpremta, \phi_{s}}{ \phi_{\exists} \lgcor \phi_{\forall}}}
\RightLabel{($\forall_{rop3}$)}
\BinaryInfC{\sqjudgement{\sqpremta, \phi_{s}}{\tfoAr{\psi_1}{\psi_{2}}}}
\end{prooftreepf}

Notice that due to the fresh placeholder generated by $\vee_{s_r}$, that the placeholder rule is similar to the rule without the placeholder.

Also notice that the implementation complexity of $\exists_{r2}$ is placed in the third premise of ($\forall$ ph1):  \sqjudgement{succ(\sqpremta), pred(\phi_{s_1})}{succ(\sqpremta), \phi_{s_2}}. The catch is to make $\phi_{s_1}$ as large as possible such that all the proof rules go through. yet ensuring that all of $pred(\phi_{s_1})$ satisfies $\psi_{2}$.

\begin{lemma}
The optimized proof rules for $\tfoAr{\psi_1}{\psi_2}$ are sound and complete.
\label{lem:derivforall2}
\end{lemma}

\begin{proof}[Proof of Lemma \ref{lem:derivforall2}]
Given the correctness of $\vee_{s_r}$, the soundness and completeness of the proof rules with and without placeholders are the same. Hence, we only give the soundness and completeness proof rules when placeholders are not used. Note that the soundness and completeness heavily depends on the derivation and the correctness of $\exists_{rw}$; we only need argue the changes to the derivation.

\textbf{Soundness:} The soundness of $\forall_{ro1}$ and $\forall_{ro2}$ follow from the definition of $\tfoAr{\psi_1}{\psi_2}$. In the first case, $\psi_2$ is always true (relativization not needed) and in the second rule, $\psi_1$ is immediately satisfied. Hence we examine rule $\forall_{ro3}$.
Suppose that the premises are true. Given that the derived proof rule is correct, we only need show that the proof rule \sqjudgement{succ(\sqpremta, \phi_{\forall})}{succ(\sqpremta) \lgcand \phi_{s_2}} results in the correct placeholders $\psi_{s_1}$ and $\phi_{s_2}$. We then invoke the soundness of the derived proof. Since we have \sqjudgement{succ(\sqpremta), pred(\phi_{s_1})}{succ(\sqpremta), \phi_{s_2}}, we know that by definition of $\subseteq$ and $\sqturn$ and the premise $\sqjudgement{succ(\sqpremta), \phi_{s_2}}{\psi_{s_2}}$ that $\sqjudgement{succ(\sqpremta), pred(\phi_{s_1})}{\psi_{2}}$. Now we have all the premises from the derived rule (including the rule $\exists_{rw}$). Hence, the proof is sound.

\textbf{Completeness:} The proof rules $\forall_{ro1}$ and $\forall_{ro2}$ cover the corner cases when $\phi_{s_1}$ is either empty or all possible clock values. We now address the completeness using rule $\forall_{ro3}$. Suppose that the conclusion is true; that \sqjudgement{\sqpremta}{\tfoAr{\psi_1}{\psi_2}}. From the soundness and completeness of the derived proof rules, the only premise that is different is \sqjudgement{succ(\sqpremta, \phi_{\forall})}{succ(\sqpremta) \lgcand \phi_{s_2}}. Using the derived proof rules (including rule $\exists_{rw}$), we get two placeholders based on $\psi_{s_2}$: $\sqjudgement{succ(\sqpremta), pred(\phi_{s_1})}{\psi_{s_2}}$ and $\sqjudgement{succ(\sqpremta), \phi_{s_2}}{\psi_{s_2}}$. Using soundness and completeness, choose $\phi_{s_2}$ to be the largest such placeholder (since the proof rules can be solved exactly with unions of clock zones, one such largest placeholder exists). Since $\phi_{s_1}$ and $\phi_{s_2}$ are clock constraints (unions of clock zones independent of location), The choice of $pred(\phi_{s_1})$ and $\phi_{s_2}$ does not change the discrete state, and only the clock state. Since both clock states are contained in the set of clock states that satisfy $\psi_{2}$, both depend on clock constraints. Since $\phi_{s_2}$ was chosen to be the largest possible such placeholder when intersection with $succ(\sqpremta)$, we know  \sqjudgement{succ(\sqpremta), pred(\phi_{s_1})}{succ(\sqpremta), \phi_{s_2}}. (Note that the proof requires the intersection with $succ(\sqpremta)$ on both sides.) Hence, we have found placeholders that satisfy all the premises of this proof rule. \qed
\end{proof}

\subsection{Derived Rules Involving Invariants}
\label{ss:a:invderiv}

We represent an invariant with $Inv$. To handle invariants, we add them to the formula (similar to how invariants are handled when they are converted to a PES in \cite{zhang-fast-generic-2005,zhang-fast-on-the-fly-2005}). We incorporate invariants into MES as follows:
$\tfoEsh{\psi} \text{ becomes } \tfoEsh{Inv \lgcand \psi}$ and
$\tfoAsh{\psi} \text{ becomes } \tfoAsh{\neg Inv \lgcor \psi}$.
These encodings require that the invariants are \emph{past closed}: if the invariant is true at some time, then the invariant must be true at all previous times. Full derived proof rules guiding the implementation of invariants for the time advance operators discussed in the remainder of this Section. Using these derived proof rules, we can reduce computation by specializing the proof tree by substituting in $Inv$ (or $\neg Inv$) for the relevant placeholders. Since we know the value of $Inv$, which is a specific clock constraint, rather than using the general-purpose rules to solve for the placeholders, we input in these values and \emph{specialize} the rules.  Invariants checked in the satisfaction of action operators $\mmEX{\psi}$ and $\mmAX{\psi}$ are handled as guards are handled in \cite{zhang-fast-generic-2005,zhang-fast-on-the-fly-2005}.

To optimize the implementation of invariants, we derive the rules that are formed when using invariants. Here, we let $Inv$ represent the invariant. Now using the invariant, we derive the proof rules:

\begin{prooftreepf}
\AxiomC{\framebox{\sqjudgement{succ(\sqpremta), \phi_{s}}{Inv}}}
\AxiomC{\sqjudgement{succ(\sqpremta), \phi_{s}}{\psi}}
\RightLabel{($\wedge$)}
\BinaryInfC{\sqjudgement{succ(\sqpremta), \phi_{s}}{Inv \lgcand \psi}}
\AxiomC{\sqjudgement{\sqpremta}{\phi_{s}}}
\RightLabel{($\exists_{t1}$)}
\BinaryInfC{\sqjudgement{\sqpremta}{\tfoEsh{Inv \lgcand \psi}}}
\end{prooftreepf}

From this derivation, we know that for $\exists$, the proper thing to do concerning invariants is to intersect the invariants with the placeholder $\phi_{s}$ and not with $succ(\sqpremta)$.

\begin{prooftreepf}
\AxiomC{\framebox{\sqjudgement{succ(\sqpremta), \phi_{s}}{Inv}}}
\AxiomC{\sqjudgement{succ(\sqpremta), \phi_{s}}{\psi}}
\RightLabel{($\wedge$)}
\BinaryInfC{\sqjudgement{succ(\sqpremta), \phi_{s}}{Inv \lgcand \psi}}
\AxiomC{\sqjudgement{\sqpremta, \phi_{\exists}}{\phi_{s}}}
\RightLabel{($\exists_{t2}$)}
\BinaryInfC{\sqjudgement{\sqpremta, \phi_{\exists}}{\tfoEsh{Inv \lgcand \psi}}}
\end{prooftreepf}

This above derivation for $\exists$ is similar to the one without the placeholder. Hence, we know to intersect the invariants with the placeholder $\phi_{s}$ and not with $succ(\sqpremta)$.

Now we derive the invariant for $\forall$:

\begin{prooftreepf}
\AxiomC{\sqjudgement{succ(\sqpremta), \phi_{s}}{\neg Inv}}
\AxiomC{\sqjudgement{succ(\sqpremta), \neg \phi_{s}}{\psi}}
\RightLabel{($\vee_{c}$)}
\BinaryInfC{\sqjudgement{succ(\sqpremta)}{\neg Inv \lgcor \psi}}
\RightLabel{($\forall_{t1}$)}
\UnaryInfC{\sqjudgement{\sqpremta}{\tfoAsh{\neg Inv \lgcor \psi}}}
\end{prooftreepf}

Since $\phi_{s} = \neg Inv$, this rule reduces to:

\begin{prooftreepf}
\AxiomC{\sqjudgement{succ(\sqpremta), Inv}{\psi}}
\RightLabel{($\vee_{c}$, derived)}
\UnaryInfC{\sqjudgement{succ(\sqpremta)}{\neg Inv \lgcor \psi}}
\RightLabel{($\forall_{t1}$)}
\UnaryInfC{\sqjudgement{\sqpremta}{\tfoAsh{\neg Inv \lgcor \psi}}}
\end{prooftreepf}

From this derivation, since there is no placeholder, we intersect the invariant $Inv$ with $succ(\sqpremta)$. Now we consider $\forall$ with a placeholder:

\begin{prooftreepf}
\AxiomC{\sqjudgement{succ(\sqpremta), \phi_{s}}{\neg Inv}}
\AxiomC{\sqjudgement{succ(\sqpremta), \phi_{\vee} - \phi_{s}}{\psi}}
\RightLabel{($\vee_{c}$)}
\BinaryInfC{\sqjudgement{succ(\sqpremta), \phi_{\vee}}{\neg Inv \lgcor \psi}}
\UnaryInfC{$\forall_{t2}$ inv1}
\end{prooftreepf}

\begin{prooftreepf}
\AxiomC{see $\forall_{t2}$ inv1}
\UnaryInfC{\sqjudgement{succ(\sqpremta), \phi_{\vee}}{\neg Inv \lgcor \psi}}
\AxiomC{\sqjudgement{succ((\sqpremta, \phi_{\forall}))}{succ(\sqpremta) \lgcand \phi_{\vee}}}
\RightLabel{($\forall_{t2}$)}
\BinaryInfC{\sqjudgement{\sqpremta, \phi_{\forall}}{\tfoAsh{\neg Inv \lgcor \psi}}}
\end{prooftreepf}

This derivation is more complex, due to the placeholder in the $\forall$. Here we use the version of $\vee_{c}$ that uses the complement of the placeholder. After doing some set operations and solving the left $(\phi_{s} = \neg Inv)$, we get the cleaned up version of $\forall$ with a placeholder:

\begin{prooftreepf}
\AxiomC{\framebox{\sqjudgement{succ(\sqpremta), Inv, \phi_{\vee}}{\psi}}}
\RightLabel{($\vee_{c}$ derived)}
\UnaryInfC{\sqjudgement{succ(\sqpremta), \phi_{\vee}}{\neg Inv \lgcor \psi}}
\AxiomC{\sqjudgement{succ((\sqpremta, \phi_{\forall}))}{succ(\sqpremta) \lgcand \phi_{\vee}}}
\RightLabel{($\forall_{t2}$)}
\BinaryInfC{\sqjudgement{\sqpremta, \phi_{\forall}}{\tfoAsh{\neg Inv \lgcor \psi}}}
\end{prooftreepf}

This means that the Invariant is not part of the placeholder, $\phi_{\vee}$. Hence, we still intersect $Inv$ with $succ(\Gamma)$, and not the placeholder, and we might to allow valuations that do not satisfy the invariant in the placeholder $\phi_{\vee}$. Furthermore, to get the largest placeholder $\phi_{\vee}$, we have to include all of $\neg Inv$ in $\phi_{v}$. This means that to get all possible valuations for the placeholders, we union the complement of the invariant ($\neg Inv$) with the placeholder $\phi_{\vee}$ and then use $\phi_{v}$ to find $\phi_{\forall}$.

To illustrate this point, we derive the $\forall$ rule with a placeholder using the $\vee_{s}$ proof rule from this paper. The alternative derivation is: 

\begin{prooftreepf}
\AxiomC{\sqjudgement{succ(\sqpremta), \phi_{s_1}}{\neg Inv}}
\AxiomC{\sqjudgement{succ(\sqpremta), \phi_{s_2}}{\psi}}
\AxiomC{\sqjudgement{succ(\sqpremta), \phi_{\vee}}{\phi_{s_1} \lgcor \phi_{s_2}}}
\RightLabel{($\vee_{s}$)}
\TrinaryInfC{\sqjudgement{succ(\sqpremta), \phi_{\vee}}{\neg Inv \lgcor \psi}}
\UnaryInfC{Inv2}
\end{prooftreepf}

\begin{prooftreepf}
\AxiomC{see Inv2}
\UnaryInfC{\sqjudgement{succ(\sqpremta), \phi_{\vee}}{\neg Inv \lgcor \psi}}
\AxiomC{\sqjudgement{succ((\sqpremta, \phi_{\forall}))}{succ(\sqpremta) \lgcand \phi_{\vee}}}
\RightLabel{($\forall_{t2}$)}
\BinaryInfC{\sqjudgement{\sqpremta, \phi_{\forall}}{\tfoAsh{\neg Inv \lgcor \psi}}}
\end{prooftreepf}

Using that $\phi_{s_1} = \neg Inv$, we get the cleaned up rule for $\forall$ with a placeholder:

\begin{prooftreepf}
\AxiomC{\framebox{\sqjudgement{succ(\sqpremta), \phi_{s}}{\psi}}}
\AxiomC{\framebox{\sqjudgement{succ(\sqpremta), \phi_{\vee}}{\phi_{s} \lgcor \neg Inv}}}
\RightLabel{($\vee_{s}$ derived)}
\BinaryInfC{\sqjudgement{succ(\sqpremta), \phi_{\vee}}{\neg Inv \lgcor \psi}}
\UnaryInfC{Inv3}
\end{prooftreepf}

\begin{prooftreepf}
\AxiomC{see Inv3}
\UnaryInfC{\sqjudgement{succ(\sqpremta), \phi_{\vee}}{\neg Inv \lgcor \psi}}
\AxiomC{\sqjudgement{succ((\sqpremta, \phi_{\forall}))}{succ(\sqpremta) \lgcand \phi_{\vee}}}
\RightLabel{($\forall_{t2}$)}
\BinaryInfC{\sqjudgement{\sqpremta, \phi_{\forall}}{\tfoAsh{\neg Inv \lgcor \psi}}}
\end{prooftreepf}

which is the same as the previous derivation using the rule $\vee_{c}$.

How we included the invariant $Inv$ in the proof rules depends on the definition of the time advance operators. Also note that these uses of $Inv$ require that $Inv$ is \textbf{past closed}.
Note that $Inv \lgcif \psi$ is equivalent to $\neg Inv \lgcor \psi$; since $\lgcif$ is not fully supported, we encode the invariant with $(\neg Inv) \lgcor \psi$ and use the derived results to handle the negation over the invariant.

\subsection{Simplified TCTL Formulas}
\label{ss:a:simplertctl}

If $\phi_1$ is an atomic proposition, conjunction, or disjunction of them (it has no fixpoint variables, transitions, time advances or clock constraints), the the relativized formulas can be simplified. Let the conjunction and disjunctive constraint (normal form not required) of atomic propositions be $p_{p}$. We can construct $p_{p}$ with the following grammar:  
\begin{align}
p_{p} ::= p \barsep \lgnot p \barsep \lgtrue \barsep \lgfalse \barsep p_{p} \lgcand p_{p} \barsep p_{p} \lgcor p_{p} \label{eq:gram}
\end{align}
where $p \in \pwmset{L}$ is an atomic proposition.

We represent such atomic literals with $p$ and $q$. If we only consider subformulas with this specified grammar, we also give simplified formulas for common TCTL operators 

\begin{theorem} Let $p$ and $q$ be a combinations of conjunctions and disjunctions of atomic propositions constructed using Equation \ref{eq:gram}. Then we have the following simplified TCTL formulas:
\begin{align}
\ctlAG{p} \equiv \ &Y \mesnueq p \lgcand \tfoAsh{\mmAX{Y}} \\
\ctlAF{p} \equiv \ &Y \mesmueq p \lgcor \Bigl(\tfoAsh{\mmAX{Y}} \lgcand \tfoEsh{\tfreeze{z}{\tfoAsh{z < 1}}}\Bigr) \\
\ctlEF{p} \equiv \ &Y \mesmueq p \lgcor \tfoEsh{\mmEX{Y}} \\
\ctlEG{p} \equiv \ &Y \mesnueq p \lgcand \Bigl(\tfoEsh{\mmEX{Y}} \lgcor \tfoAsh{\tfreeze{z}{\tfoEsh{z \geq 1}}}\Bigr)\\
\ctlAG{p \lgcif \ctlAF{q}} \equiv \ & Y \mesnueq (\neg p) \lgcor \tfoAsh{Y_2 \lgcand \mmAX{Y}}\nonumber \\ 
											 & Y_2 \mesmueq q \lgcor (\tfoAsh{\mmAX{Y_2}} \lgcand \tfoEsh{\tfreeze{z}{\tfoAsh{z < 1}}}) 
\end{align}
\label{lem:simpltctl}
\end{theorem}

The last operator is the ``leads to'' operator. Here we use the simplified $\ctlAF{q}$ but use the regular \ctlAG{p} formula. Also recall that the tool has operators to handle the subpaths with the freeze quantifiers.

To prove this operators, we will rely on some of the properties of formulas involving only atomic propositions. The proofs rely on the following property: if $p$ is true, then $\tfoAsh{p}$ is true. Also, for all previous times, $p$ is true. Hence, the semantics of the formula is equivalent regardless of whether a continuous or a pointwise semantics is used for $p$. As a result, we have the equivalences in the following lemma:

\begin{lemma}[Properties of atomic proposition formulas] Let $p$ be a combination of conjunctions and disjunctions of atomic propositions. Then:
\begin{align}
&p \equiv \tfoEsh{p} \equiv \tfoAsh{p} \\
&p \lgcor \tfoAsh{\phi}  \equiv \tfoAr{p}{p \lgcor \phi} \text{ for any formula $\phi$} \\
&p \lgcand \tfoEsh{\phi} \equiv  \tfoEr{p}{p \lgcand \phi}  \text{ for any formula $\phi$}
\end{align}
\label{lem:propequiv}
\end{lemma}

\begin{proof}[Proof of Lemma \ref{lem:propequiv}]
From the definitions of the $L^{rel}_{\nu,\mu}$ operators. For some insight into the second and third equivalences, try using $\phi = \lgfalse$ and $\phi = \lgtrue$. \qed
\end{proof}

\begin{proof}[Proof of Theorem \ref{lem:simpltctl}] 
Here we show $\ctlAG{p}$ and $\ctlAF{p}$. The proofs for $\ctlEG{p}$ and $\ctlEF{p}$ are similar, and the proof of the last equivalence follows from the proofs for $\ctlAG{p}$ and $\ctlAF{p}$.

Proof of $\ctlAG{p}$:
\begin{align*}
\ctlAG{p} \equiv  & Y \mesnueq \tfoAsh{p \lgcand \mmAX{Y}} \text{ (Original Formula)} \\
	& Y \mesnueq \tfoAsh{p} \lgcand \tfoAsh{\mmAX{Y}} \text{ (Distributivity $\forall$, $\wedge$)} \\
	& Y \mesnueq p \lgcand \tfoAsh{\mmAX{Y}} \text{ ($p \equiv \tfoAsh{p}$)}
\end{align*}

Proof of $\ctlAF{p}$:
\begin{align*}
\ctlAF{p} \equiv &  
	Y \mesmueq \tfoArB{p}{p \lgcor \mmAX{Y}} \lgcand 
	  \Bigl(\tfoEsh{\tfreeze{z}{\tfoAsh{z < 1}}} \lgcor \tfoEsh{p}\Bigr) \text{ (Original Formula)} \\
	  &Y \mesmueq (p \lgcor \tfoAsh{\mmAX{Y}}) \lgcand 
	  \Bigl(\tfoEsh{\tfreeze{z}{\tfoAsh{z < 1}}} \lgcor \tfoEsh{p}\Bigr) \text{ ($p \lgcor \tfoAsh{\phi}  \equiv \tfoAr{p}{p \lgcor \phi}$)} \\
	  & Y \mesmueq (p \lgcor \tfoAsh{\mmAX{Y}}) \lgcand 
	  \Bigl(\tfoEsh{\tfreeze{z}{\tfoAsh{z < 1}}} \lgcor p \Bigr) \text{ ($p \equiv \tfoEsh{p}$)} \\
	  & Y \mesmueq (p \lgcor \tfoAsh{\mmAX{Y}}) \lgcand 
	  (p \lgcor  \tfoEsh{\tfreeze{z}{\tfoAsh{z < 1}}}) \text{ (Commutativity $\vee$)} \\
	  &  Y \mesmueq p \lgcor \Bigl(\tfoAsh{\mmAX{Y}} \lgcand  \tfoEsh{\tfreeze{z}{\tfoAsh{z < 1}}}\Bigr) \text{ (Distributivity $\wedge$, $\vee$)}
\end{align*}  \qed
\end{proof}

\section{Experiment Data and Analysis}
\label{s:a:experiment}

\subsection{Methods: Evaluation Models}
\label{ss:a:models}

In our case study, we use four different models: Carrier Sense, Multiple Access with Collision Detection (CSMA); Fischer's Mutual Exclusion (FISCHER); Generalized Railroad Crossing (GRC); and Leader election (LEADER). These models were used in and taken from \cite{heitmeyer-the-generalized-1994,zhang-fast-generic-2005,zhang-fast-on-the-fly-2005}. Here is a brief description of them:

\begin{enumerate}
	\item \emph{Carrier Sense, Multiple Access with Collision Detection (CSMA).} There are $n$ processes sharing who one bus.  The bus can only send one message at a time.  At various times processes will try to transmit a message.  If the process detects that the bus is busy, then the process will wait a random amount of time before retrying. 
	\item \emph{Fischer's Mutual Exclusion (FISCHER)} This protocol involves $n$ processes vying for access to a critical section.  Each process asks for the critical section and then waits until it gets it, re-requesting for access if it is not granted it for a period of time. The critical section identifies which process currently has access to it.
	\item \emph{Generalized Railroad Crossing (GRC).} This protocol has $n$ trains, a gate and a controller. The trains cross a region that intersects a road, and the gate goes down to prevent cars from driving on the road when a train is passing through. When no train is nearby, the gate raises or remains up to allow cars to safely drive through.
	\item \emph{Leader election (LEADER).}  This protocol involves involves $n$ processes that are electing a leader amongst themselves.  To elect a leader, at each step one process asks another process to be its parent. In our model, the smaller-numbered process always becomes the parent. When finished, the process with no parent is the leader. 
\end{enumerate}

\subsection{Methods: Evaluation Specifications}
\label{ss:a:specs}

The specifications that are not supported by UPPAAL are in \emph{italics} and are marked with a $^{*}$.

The specifications checked on the CSMA protocol are:
\begin{itemize}
	\item \emph{\textbf{AS$^{*}$:}  At most one process is in a transmission state for less than 52 ($2\sigma$) units. (Valid)}
	\item \textbf{BS:} At any time, a third process can retry while two are already in transmission status. (Invalid)
	\item \textbf{AL:}  It is inevitable that all processes are waiting. (Valid)
	\item \textbf{BL:} It is inevitable that some process needs to retry transmitting a message. (Invalid)
	\item \textbf{M1:} It is always the case that if the first process needs to retry that it will inevitably transmit. (Invalid) 
	\item \textbf{M2:} It is always the case that if a bus experiences a collision that it will inevitably become idle. (Valid)
	\item \emph{\textbf{M3$^{*}$:} The bus is always idle until a process is active. (Invalid)} 
	\item \emph{\textbf{M4$^{*}$:} For all paths with an infinite number of actions, the bus is always idle until a process is active (Valid)}
\end{itemize}

The specifications checked on the FISCHER protocol are:
\begin{itemize}
	\item \textbf{AS:}  At any time, at most one process is in the critical section. (Valid)
	\item \textbf{BS:} At any moment, at most four processes in their waiting state at the same time. (Valid for four processes, Invalid for five or more processes)
	\item \textbf{AL:} It is inevitable that all processes are idle. (Valid)
	\item \textbf{BL:} It is inevitable that some process accesses the critical section. (Invalid)
		\item \textbf{M1:} It is always the case that if the first process is not idle, it will eventually access the critical section. (Invalid)
	\item \textbf{M2:} It is always the case that if the third process is not idle, it will eventually access the critical section. (Invalid)
	\item \emph{\textbf{M3$^{*}$:} It is possible for the first process to enter the critical section without waiting. (Invalid)}
	\item \emph{\textbf{M4$^{*}$:} After at most five action transitions, some process will enter the critical section. (Invalid)}
\end{itemize}

The specifications checked on the GRC protocol are:
\begin{itemize}
	\item \textbf{AS:} It is always the case that if at least one train (process) is in the track region, the gate is always down. (Valid)
	\item \textbf{BS:}  It is always the case that if the gate is raising then the controller (when one train is approaching or in) will not want to lower the gate. (Invalid)
	\item \textbf{AL:} It is inevitable that the gate is up. (Valid)
	\item \textbf{BL:} It is inevitable that the train is near the gate. (Invalid)
	\item \textbf{M1:} It is always the case that if the gate is down, then it will inevitably come up (Invalid).
		\item \emph{\textbf{M2$^{*}$:} It is always the case that if the gate is down, then it will inevitably come up after 30 seconds (Invalid).}
		\item \textbf{M3:} It is always the case that at most one train is in the region at one time (Invalid).
		\item \emph{\textbf{M4$^{*}$:} For all paths with an infinite number of actions, the gate is up until a train approaches (Valid).}
		\item \emph{\textbf{M4ap$^{*}$:} For all paths, the gate is up until a train approaches (Invalid).} 
\end{itemize}

The specifications checked on the LEADER protocol are:
\begin{itemize}
	\item \textbf{AS:} At any time, each process either has no parent or has a parent with a smaller process id (and thus the first process has no parent at all times). (Valid)
	\item \textbf{BS:} At any moment, at least three processes do not have parents.	(Invalid)
	\item \textbf{AL:}  It is inevitable that the first process is elected the leader. (Valid)
	\item \textbf{BL:} It is inevitable that the third processes' parent is the second process. (Invalid)
		\item \emph{\textbf{M1$^{*}$:} For all paths, a the second process cannot have a child until it has a parent. (Invalid)}
		\item \textbf{M2:} It is always the case that if the third process is assigned a parent (chosen to not be leader), then it will not be the leader. (Valid)
		\item \emph{\textbf{M3$^{*}$:} It is possible that it takes longer than 3 time units to elect a leader. (Valid)}
		\item \emph{\textbf{M4$^{*}$:} For all paths, in at most three votes, a leader is elected.  (Valid for four or fewer processes, invalid for five or more processes.)}  
\end{itemize}

\subsection{Data and Results}
\label{ss:a:data}

A scatter plot of the data in Table \ref{tab:data1} is given in Appendix \ref{ss:a:data}. In that figure, any example with O/M, TO or TOsm had its time set to 7200s (2 hours).

\begin{figure}[htb]
\includegraphics[width=0.7\linewidth]{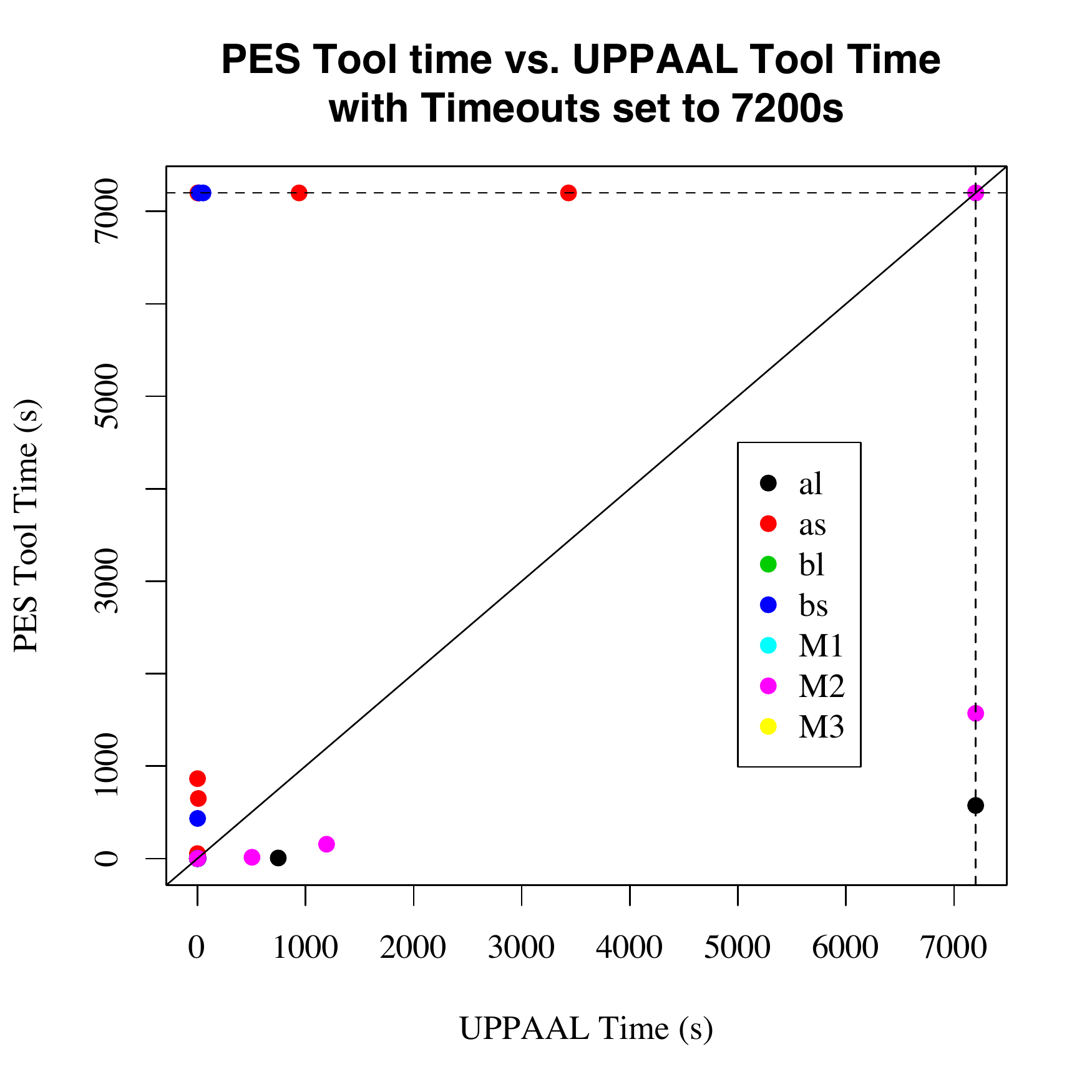}
\caption{Figure comparing the PES tool time performance with UPPAAL time performance. Points are colored by the specification category. All timed out (TO) examples or examples that ran out of memory (O/M) have their time set to 7200s, the value of the dashed lines.}
\label{fig:scplot}
\end{figure}

\end{document}